%% file: main.tex
\documentclass[10pt,twocolumn,twoside]{IEEEtran}
\input{preamble/pkgs}
\input{preamble/defs}

\begin{document}

\title{\input{preamble/title}}
\input{preamble/authors}
\maketitle
\input{sections/abs}
\input{sections/keywords}

\section{Introduction} \label{Section:Introduction}
\input{sections/intro/introduction}

\section{Preliminaries} \label{Section:Preliminaries}
\input{sections/pre/preliminary}

\section{Disentangled Control} \label{Section:DisentangledControl}
\input{sections/disentangled/disco}

\section{Applications and Experiments} \label{Section:Applications}
\input{sections/AppExp/app}

\section{Conclusion} \label{Section:Conclusion}
\input{sections/concl}

\appendices
\section{Limitations of Related Works} \label{Section:appendix}
\input{sections/appendix1}

\bibliographystyle{IEEEtran}
\bibliography{references}{}

\end{document}

%% file: preamble/pkgs.tex
\usepackage{amsmath,amsfonts}
\usepackage{algorithmic}
\usepackage{algorithm}
\usepackage{array}
\usepackage{booktabs}
\usepackage{lipsum}
\usepackage{ragged2e}
\usepackage{textcomp}
\usepackage{stfloats}
\usepackage{url}
\usepackage{verbatim}
\usepackage{graphicx}
\usepackage{relsize}
\usepackage{cite}
\usepackage{subfigure}
\usepackage{bm}
\usepackage[english]{babel}
\usepackage{amsthm}
\usepackage{color}
\usepackage{enumerate}
\usepackage{amssymb}
\usepackage{mathrsfs}
\usepackage{balance}
\usepackage[hidelinks]{hyperref}
\usepackage{mathtools}

%% file: preamble/defs.tex
\DeclareMathOperator*{\argmin}{arg\,min}
\newtheorem{theorem}{Theorem}[section]
\newtheorem{definition}{Definition}[section]
\newtheorem{lemma}{Lemma}[section]

\newtheorem{proposition}{Proposition}[section]

%% file: preamble/title.tex
\title{Disentangled Control of Multi-Agent Systems}

%% file: preamble/authors.tex
\author{
Ruoyu Lin$^{1}$,
Gennaro Notomista$^{2}$,
and Magnus Egerstedt$^{3}$
\thanks{This work was supported in part by the U.S. Army Research Lab through ARL DCIST CRA W911NF-17-2-0181, and in part by the NSERC Discovery under Grant RGPIN-2023-03703.}
\thanks{$^{1}$Ruoyu Lin is with the Department of Electrical Engineering and Computer Science, University of California, Irvine, Irvine, CA 92697, USA. Email: {\tt\small \href{rlin10@uci.edu}{\tt\small rlin10@uci.edu}}}
\thanks{$^{2}$Gennaro Notomista is with the Department of Electrical and Computer Engineering, University of Waterloo, Waterloo, ON, Canada. Email: {\tt\small \href{mailto:gennaro.notomista@uwaterloo.ca}{\tt\small gennaro.notomista@uwaterloo.ca}}}
\thanks{$^{3}$Magnus Egerstedt is with the University of North Carolina at Chapel Hill, Chapel Hill, NC, 27599, USA. Email: {\tt\small \href{magnus@unc.edu}{\tt\small magnus@unc.edu}}}
}

%% file: sections/abs.tex
\begin{abstract}
This paper develops a general framework with convergence guarantees for multi-agent control synthesis, which applies to a wide range of problems, including those with time-varying objective functions. The proposed framework achieves decentralization without inducing entangled dynamics among agents, and it naturally supports multi-objective robotics and real-time implementation. To demonstrate its generality and effectiveness, the framework is applied to three representative problems, namely time-varying leader-follower formation control, decentralized coverage control for time-varying density functions, which is a long-standing open problem, and safe formation navigation in a dense environment.
\end{abstract}

%% file: sections/keywords.tex
\begin{IEEEkeywords}
Disentangled control, networks of autonomous agents, cooperative control, robotics.
\end{IEEEkeywords}

%% file: sections/intro/introduction.tex
Multi-agent systems are increasingly used to tackle problems across energy networks, transportation systems, robotics, reinforcement learning, and generative modeling (see, e.g., \cite{bullo2009distributed,mesbahi2010graph,olfati2007consensus,bidram2013secondary,dresner2008multiagent,tan1993multi,ghosh2018multi}). Their strength lies in shaping interactions and distributing decision making, policy synthesis, or action generation among autonomous agents, enabling resilience to individual failures, effective collaboration, and adaptivity to changing environments. At the heart of these capabilities is decentralization, which also helps ensure scalability by avoiding a prohibitive increase in per-agent computational load as the number of agents grows.

While multi-agent systems have been adopted in a wide variety of applications, many theoretical developments have been informed by ideas from multi-agent control, through canonical problems such as formation control \cite{egerstedt2001formation}, coverage control \cite{cortes2004coverage}, caging control \cite{pereira2004decentralized}, containment control \cite{ji2008containment}, and consensus control (e.g., rendezvous, flocking, and cyclic pursuit) \cite{cortes2017coordinated}. Although these problems differ in their specific objectives, they share common underlying structures such as local interaction rules and coherent system-level behaviors. This observation leads to the question: Can one develop a general framework that captures these diverse tasks within a unified formulation? Such a framework may not only advance a new understanding of decentralized multi-agent control, but also inspire extensions to other fields of research.
\input{FigTab/Fig1}

The gradient flow method has been explored as a general framework for multi-agent control (see, e.g., \cite{schwager2009gradient,cortes2017coordinated}), achieving decentralization with convergence guarantees by leveraging LaSalle’s invariance principle. However, such guarantees do not necessarily hold in time-varying settings. In addition, a constraint-driven method developed in \cite{notomista2019constraint} provides a general framework for multi-agent control. However, there are several aspects not addressed by \cite{notomista2019constraint}, particularly when the underlying graph exhibits certain topologies or when the associated cost function is time-varying, as discussed in detail in Appendix~\ref{Section:appendix}.

Another challenge arises from certain forms of dynamics interdependency in multi-agent systems. Specifically, rather than acting based on its own state and those of its neighbors, each agent cannot make its decisions or determine its actions without explicitly knowing those of its neighbors, while these neighbors face the same dilemma simultaneously (see, e.g., \cite{farina2012distributed,trodden2017distributed,lin2025heterogeneous,lee2015multirobot,stewart2010cooperative,butler2024collaborative}), such as the subsystem enclosed by the pink shaded area illustrated in Fig.~\ref{fig:Entangled_graph}. Existing work has sought to address this type of challenge through approaches such as modeling neighbor dynamics as bounded disturbance within a robust control framework \cite{farina2012distributed,trodden2017distributed}, zero-order-hold approximation of neighbor dynamics \cite{lin2025heterogeneous}, centralized aggregation followed by decentralized approximation \cite{lee2015multirobot}, and iterative negotiation toward a coherent solution \cite{stewart2010cooperative,butler2024collaborative}. Although these methods have enabled practical advances, their applicability often involves trade-offs among performance conservatism, computational complexity, and provable convergence guarantees. An exact decentralization mechanism without reliance on the aforementioned approximations has remained elusive.

These considerations motivate a framework that is general enough to accommodate diverse multi-agent control tasks, while also being able to deal with time-varying cost functions and synthesize controllers in a decentralized manner. The main contributions of this paper are summarized as follows.
\input{FigTab/Tab1}

The remainder of this paper proceeds as follows. In Section~\ref{Section:Preliminaries},  we present necessary preliminaries used throughout the paper. In Section~\ref{Section:DisentangledControl}, first, we identify entangled dynamics of multi-agent systems, which causes the key challenge of decentralized control synthesis. We also detail new insights into the closely related works \cite{notomista2019constraint} and \cite{santos2019decentralized} in Appendix~\ref{Section:appendix}. Then, we propose disentangled control, a general multi-agent control framework that resolves this challenge with convergence guarantees. The effectiveness of disentangled control is demonstrated across three experiments in Section~\ref{Section:Applications}. Finally, Section~\ref{Section:Conclusion} concludes the paper and highlights the potential of disentangled control for broader applications.

%% file: FigTab/Fig1.tex
\begin{figure}[t]
\centering
\vspace{0.2cm}
\includegraphics[scale=0.21]{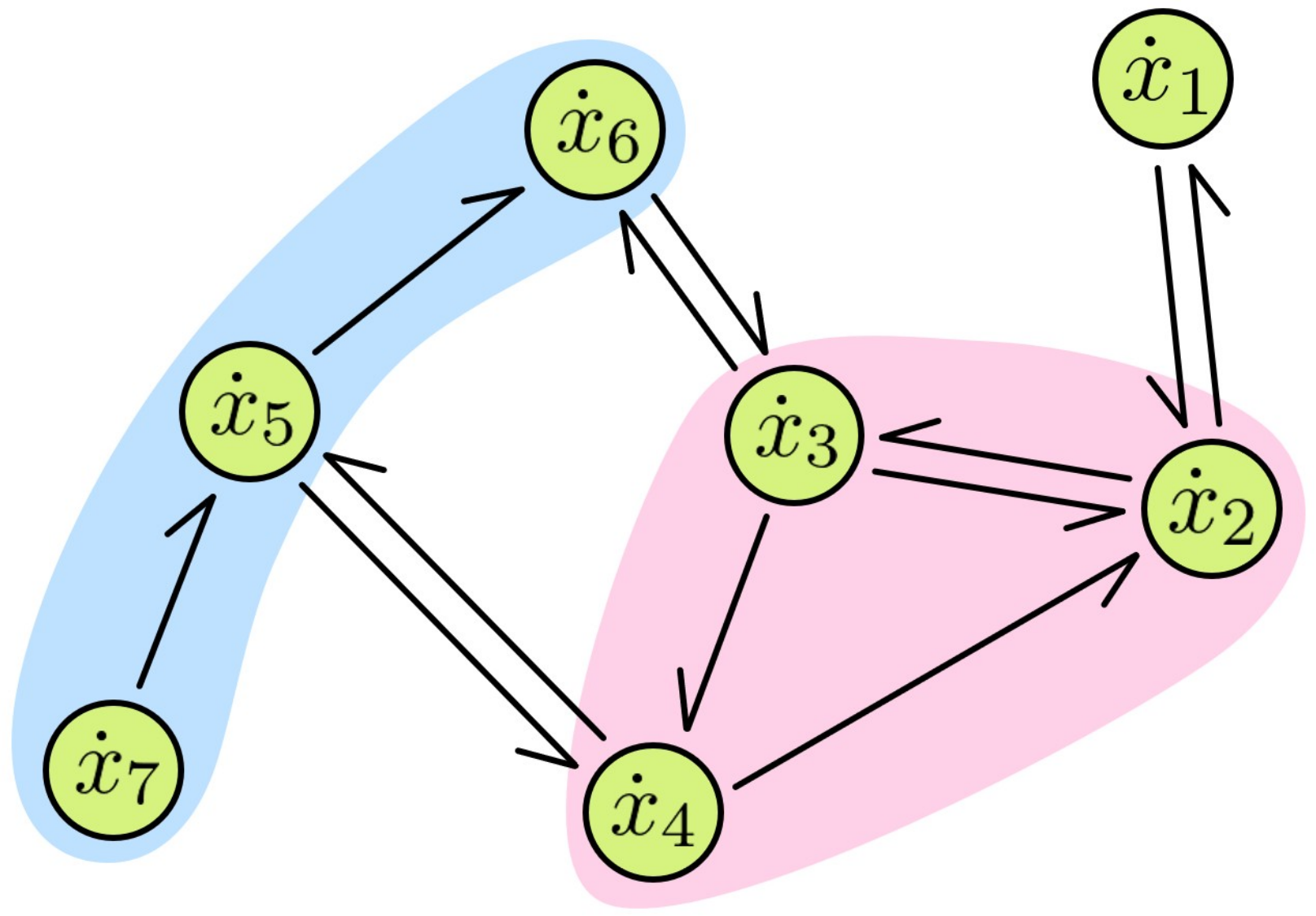}
\caption{A dynamical dependency graph illustrating the entangled dynamics of a multi-agent system. Each node corresponds to the dynamics of agent $i$, and a directed edge from node $i$ to node $j$ indicates that the dynamics of agent $j$ depends on that of agent $i$, i.e., $\frac{\partial \dot{x}_j}{\partial \dot{x}_i} \neq 0$. For example, the dynamics of the subsystem enclosed by the blue shaded area is not entangled, whereas that of the subsystem enclosed by the pink shaded area is entangled.}
\label{fig:Entangled_graph}
\end{figure}

%% file: FigTab/Tab1.tex
\begin{itemize}
\item \input{sections/contrib1}
\item \input{sections/contrib2}
\item \input{sections/contrib3}
\end{itemize}

%% file: sections/contrib1.tex
We analyze two closely related and widely adopted multi-agent control approaches, \cite{notomista2019constraint} and \cite{santos2019decentralized}, and show that there are some missing pieces not covered in \cite{notomista2019constraint} and several theoretical limitations of \cite{santos2019decentralized}.

%% file: sections/contrib2.tex
We propose disentangled control, a general framework such that entangled dynamics of multi-agent systems do not arise, thereby enabling decentralized control synthesis. Within this framework, we consider two related structures leading to entangled dynamics and develop corresponding synthesis methods with convergence guarantees. The proposed framework handles time-varying objective functions, accommodates multi-objective optimization, supports real-time computation, and is applicable to a broad range of multi-agent control problems. 

%% file: sections/contrib3.tex
The generality and effectiveness of the proposed framework are demonstrated through one simulation and two physical experiments, namely time-varying leader-follower formation control, decentralized coverage control for time-varying densities, and safe formation navigation in an environment with densely distributed obstacles. Notably, the problem of coverage control for time-varying density functions has remained a long-standing open challenge, for which, to the best of our knowledge, no decentralized control synthesis method without approximations has been established over the past two decades.

%% file: sections/pre/preliminary.tex
Since the dynamics of a wide range of systems can be expressed in a control affine form \cite{khalil2002control}, in this paper, we consider the dynamical system
\begin{equation}
\label{eqn:controlaffine}
\dot{x} = f(x) + g(x) u,
\end{equation}
where $f: \mathbb{R}^n \to \mathbb{R}^n$ and $g: \mathbb{R}^n \to \mathbb{R}^{n\times m}$ are locally Lipschitz,
and $u \in \mathcal{U} \subseteq \mathbb{R}^m$ is the control input. 

Given a real-valued function $F : \mathbb{R}^n \times \mathbb{R}_{\geq0} \to \mathbb{R}$, its total derivative along a trajectory of \eqref{eqn:controlaffine} is given by
\begin{equation*}
\frac{d}{dt} F(x, t)
= 
{\nabla_{\!x}} F(x,t) \, \dot{x}
+ 
{\partial_t}F(x, t),
\end{equation*}
in which
\begin{equation*}
{\nabla_{\!x}} F \coloneqq
\frac{\partial}{\partial x} F \in \mathbb{R}^{1\times n},
\quad
{\partial_t}F
\coloneqq
\frac{\partial}{\partial t} F \in \mathbb{R}.
\end{equation*}
Throughout this paper, for notational compactness, we sometimes also use the dot notation $\dot{(\cdot)}$ to denote the total derivative of a quantity. To distinguish the Jacobian notation from the gradient notation, for a vector-valued function $G:\mathbb{R}^n\times\mathbb{R}_{\ge 0}\to\mathbb{R}^p$, we denote its Jacobian by
\begin{equation*}
\frac{\partial G}{\partial x} \in \mathbb{R}^{p\times n}. 
\end{equation*}

The standard approach in multi-agent control is to formulate a multi-agent task in terms of certain performance cost or task objective, and to design a control law such that this cost (or objective) function is minimized \cite{schwager2009gradient,cortes2017coordinated}. 

Consider the objective of minimizing a real-valued cost function $V:\mathbb{R}^n \to \mathbb{R}_{\geq 0}$ of class $C^1$. A natural reference dynamics for this purpose is the gradient flow
\begin{equation*}
\label{eqn:GF}
\dot{x} = -{\nabla_{\!x}} V(x)^{\!\top},
\end{equation*}
as $\dot{V}(x) = {\nabla_{\!x}} V(x) \,\dot{x} = - \|{\nabla_{\!x}} V(x)\|^2 \leq 0$. The gradient flow method can be viewed as the continuous-time version of the gradient descent method (see, e.g., \cite{ambrosio2005gradient,boyd2004convex}). 

However, if the cost function depends explicitly on time, i.e., $V: \mathbb{R}^n \times \mathbb{R}_{\geq 0} \to \mathbb{R}_{\geq 0}$, the gradient flow method does not necessarily guarantee its convergence, (see, e.g., \cite{lee2015multirobot}). To deal with time-varying cost functions, one can resort to the following standard convergence lemma for differential inequalities (see, e.g., \cite{lin1996smooth,khalil2002control}).
\input{theorem/lemma1}

The problem of minimizing a time-varying cost function (TVCF) seems to be solved by Lemma~\ref{lemma:TVCLF}. However, its decentralized realization in multi-agent systems is nontrivial, which is the main focus of this paper. The main challenge arises from certain forms of dynamics
interdependency across different agents as mentioned in Section~\ref{Section:Introduction},  which is explored in the following section.

%% file: theorem/lemma1.tex
\begin{lemma}
\label{lemma:TVCLF}
Let $V:\mathbb{R}^n\times\mathbb{R}_{\ge0}\to\mathbb{R}_{\ge0}$ be a class $C^1$ function, and $\alpha:\mathbb{R}_{\ge0}\to\mathbb{R}_{\ge0}$ be a class $\mathcal{K}_{\infty}$ function. If there exists a controller $u$ such that a trajectory $x:\mathbb{R}_{\ge0}\to\mathbb{R}^n$ of \eqref{eqn:controlaffine} exists, along which
\begin{equation}
\label{eqn:cost_decay_constraint}
\frac{d}{dt} V(x(t),t) + \alpha(V(x(t),t)) \le 0,
\end{equation}
$\forall t\ge0$, then $\lim_{t\to\infty}V(x(t),t)=0$.
\end{lemma}
\input{results/proof1}

%% file: results/proof1.tex
\begin{proof}
Define $W(t) \coloneqq V(x(t),t)$. Then, by \eqref{eqn:cost_decay_constraint}, $\dot{W}(t)+\alpha(W(t))\le0$, $\forall t\geq 0$. Since $\alpha$ is of class $\mathcal K_{\infty}$, we have $\dot{W}(t)\le -\alpha(W(t))\le 0$, $\forall t\geq 0$, which implies that $W(t)$ is nonincreasing. As $W(t)\ge0$, $\forall t\geq 0$, the limit $W_{\infty} \coloneqq \lim_{t\to\infty}W(t)$ exists and satisfies $W_{\infty}\ge0$. We now show that $W_{\infty}=0$ by contradiction. Assume that $W_{\infty}>0$. Given any $W_c \in (0,W_{\infty})$, since $\alpha$ is of class $\mathcal K_{\infty}$, we have $\alpha(W_c)>0$. Thus, $\exists T\geq 0$ such that $W(t)\geq W_c$, $\forall t \geq T$. Hence, we have $\dot{W}(t)\le -\alpha(W(t))
\leq -\alpha (W_c)$, $\forall t\ge T$, leading to $W(t) \leq W(T)-(t-T)\,\alpha(W_c)\coloneqq \overline{W}_T (t)$. Thus, we have $\lim_{t\to\infty} \overline{W}_T (t) = -\infty$, which contradicts the fact that $W(t)\geq 0$, $\forall t\geq 0$. Therefore, $W_{\infty}=0$, i.e., $\lim_{t\to\infty}V(x(t),t)=0$.
\end{proof}

%% file: sections/disentangled/disco.tex
Consider a system of $N \in \mathbb{Z}^+$ agents, each of whose states of interest is denoted by $x_i \in \mathbb{R}^{n_i}$, $\forall i \in \mathcal{N} \coloneqq \{1, \dots, N\}$. The vector $x = [x_1^\top, \dots, x_N^\top]^\top \in \mathbb{R}^{\sum_{i \in \mathcal{N}}n_i}$ consists of the states of all $N$ agents. Each agent is governed by the control affine dynamics 
\begin{equation}
\label{eqn:controlaffine_ma}
\dot{x}_i = f_i(x_i) + g_i(x_i) u_i, \quad \forall i \in \mathcal{N}, 
\end{equation}
where $f_i: \mathbb{R}^{n_i} \to \mathbb{R}^{n_i}$ and $g_i: \mathbb{R}^{n_i} \to \mathbb{R}^{n_i \times m_i}$ are locally Lipschitz, and $u_i \in \mathcal{U}_i \subseteq \mathbb{R}^{m_i}$ is the control input for agent $i$.

\subsection{Entangled Dynamics of Multi-Agent Systems} \label{Section:EntangledDynamics}
\input{sections/disentangled/entangled}

\subsection{Disentangled Control of Multi-Agent Systems} \label{Section:subsecdisentangledcontrol}
\input{sections/disentangled/ctrl}

%% file: sections/disentangled/entangled.tex
Consider the multi-agent system illustrated in {\rm Fig.~\ref{fig:Entangled_graph}}. At a given time instant, suppose agent 3 aims to compute its dynamics $\dot{x}_3$. This requires the knowledge of $\dot{x}_2$ and $\dot{x}_6$ since $\frac{\partial \dot{x}_3}{\partial \dot{x}_2} \neq 0$ and $\frac{\partial \dot{x}_3}{\partial \dot{x}_6} \neq 0$. However, computing $\dot{x}_2$ in turn requires the information of $\dot{x}_1$, $\dot{x}_3$, and $\dot{x}_4$, because $\frac{\partial \dot{x}_2}{\partial \dot{x}_1} \neq 0$, $\frac{\partial \dot{x}_2}{\partial \dot{x}_3} \neq 0$, and $\frac{\partial \dot{x}_2}{\partial \dot{x}_4} \neq 0$, while $\dot{x}_1$ in turn depends on $\dot{x}_2$, and $\dot{x}_4$ itself also depends on $\dot{x}_3$ and $\dot{x}_5$. Similarly, computing $\dot{x}_6$ requires the knowledge of $\dot{x}_3$ and $\dot{x}_5$ as $\frac{\partial \dot{x}_6}{\partial \dot{x}_3} \neq 0$ and $\frac{\partial \dot{x}_6}{\partial \dot{x}_5} \neq 0$, where $\dot{x}_5$ further depends on $\dot{x}_4$ which again depends on $\dot{x}_3$. As a result, agent 3 cannot compute $\dot{x}_3$ in a decentralized manner. Instead, $\dot{x}_1, \dot{x}_2, \dot{x}_3, \dot{x}_4, \dot{x}_5$, and $\dot{x}_6$ must be solved simultaneously in a centralized fashion.

The aforementioned interdependent relationships of dynamics form a strongly coupled structure, which we refer to as \textbf{\textit{entangled dynamics}} of multi-agent systems.
By \textit{entangled dynamics} we mean that in order for an agent to obtain its dynamics, it must have access to the dynamics of certain other agents, and in turn, those agents require this agent’s dynamics (and possibly the dynamics of additional agents) to obtain their own dynamics. This creates a more complicated version of the ``chicken-and-egg'' coupling issue, as illustrated in {\rm Fig.~\ref{fig:Entangled_graph}}.

To show how such entangled dynamics naturally arises in multi-agent systems and why its decentralized realization is challenging, we discuss a representative problem that has received considerable attention in the literature, i.e., coverage control for time-varying density functions (which is also referred to as time-varying coverage control).

To start with, the (time-invariant) coverage control \cite{cortes2004coverage} concerns the problem of minimizing 
\begin{equation*}
C(x) = \sum_{i \in \mathcal{N}} \int_{\mathcal{V}_i(x)}\! 
\|x_i - q\|^2 \phi(q)
\,\mathrm{d}q,
\end{equation*}
which encodes the lack of coverage quality of a bounded convex domain of interest $\mathcal{D} \subset \mathbb{R}^d$, where
\begin{equation*}
\mathcal{V}_i(x) 
= 
\left\{ 
q \in \mathcal{D} \,\big|\, 
\|q - x_i\| 
\leq
\|q - x_j\|,\; \forall j\in \mathcal{N} \setminus\! \{i\} 
\right\}
\end{equation*}
is the Voronoi cell of agent $i$, and $\phi: \mathcal{D} \to \mathbb{R}_{\geq 0}$ is a continuous density distribution function encoding the relative importance or informativeness of a point $q \in \mathcal{D}$, which is not identically zero on $\mathcal{V}_i(x)$, $\forall i \in \mathcal{N}$. According to \cite{cortes2004coverage}, the gradient flow method (also known as the continuous-time limit of Lloyd's algorithm \cite{lloyd1982least}),
\begin{equation*}
\dot{x}_i = \kappa (G_i(x) - x_i), \quad \forall i \in \mathcal{N},
\end{equation*}
in which $\kappa > 0$, and $G_i(x) = \frac{1}{M_i(x)} \int_{\mathcal{V}_i(x)}\! q\phi(q) \,\mathrm{d}q$ with $M_i(x) = \int_{\mathcal{V}_i(x)}\! \phi(q) \,\mathrm{d}q$ is the center of mass of $\mathcal{V}_i(x)$, ensures asymptotic convergence to a critical point of $C(x)$ with respect to $x$ (guaranteed by LaSalle’s invariance principle), i.e., a centroidal Voronoi tessellation (CVT),
\begin{equation*}
x_i^{*} = G_i(x^{*}), \quad \forall i \in \mathcal{N},
\end{equation*}
in a decentralized manner.

When the density distribution function becomes time-varying, i.e., $\phi: \mathcal{D} \times \mathbb{R}_{\geq 0} \to \mathbb{R}_{\geq0}$, the Lloyd’s algorithm does not guarantee asymptotic convergence to a time-varying CVT (TV-CVT), as analyzed in \cite{lee2015multirobot}. In the time-varying setting, the corresponding locational cost function becomes
\begin{equation*}
C(x,t) = \sum_{i \in \mathcal{N}} \int_{\mathcal{V}_i(x)}\! 
\|x_i - q\|^2 \phi(q,t)
\,\mathrm{d}q,
\end{equation*}
in which $\phi(q,t)$ is continuous in $q$, continuously differentiable in $t$, and not identically zero on $\mathcal{V}_i(x)$, $\forall i \in \mathcal{N}$, $\forall t \geq 0$.
The corresponding time-varying mass and center of mass of $\mathcal{V}_i(x)$ take the forms of
\begin{align*}
M_i(x,t) &= \int_{\mathcal{V}_i(x)}\! \phi(q,t)
\,\mathrm{d}q
\end{align*}
and
\begin{align*}
G_i(x,t) &= \frac{1}{M_i(x,t)} \int_{\mathcal{V}_i(x)}\! q\phi(q,t)
\,\mathrm{d}q,
\end{align*}
respectively. However, it is challenging to ensure asymptotic convergence to a TV-CVT in a decentralized manner. In the literature, \cite{cortes2002coverage} and \cite{lee2015multirobot} establish the theoretical foundations for this problem.

Initially, \cite{cortes2002coverage} proposes the dynamics
\begin{align}
\label{eqn:cortes}
\dot{x}_i
= \dot{G}_i
- \left( \kappa + \frac{\dot{M}_i}{M_i} \right)(x_i - G_i),
\end{align}
$\forall i \in \mathcal{N}$, where $\kappa >0$ (and the arguments are omitted for notational simplicity), with certain quasi-static approximations of $\dot{G}_i$, $\dot{M}_i$, and ${\partial_t} \phi$. However, in general,
\begin{equation}
\dot{G}_i
\neq
{\partial_t} {G}_i
= \frac{1}{{M}_i}\int_{\mathcal{V}_i} \!
(q - {G}_i)\,
{\partial_t}{\phi}
\,\mathrm{d}q \notag
\end{equation}
because
\begin{equation*}
\dot{G}_i
= 
\frac{\partial G_i}{\partial x_i} \,\dot{x}_i 
+ \!
\sum_{j \in \mathcal{N}_i} \frac{\partial G_i}{\partial x_j} \,\dot{x}_j
+ 
{\partial_t}G_i,
\end{equation*}
where $\mathcal{N}_i$ denotes the Delaunay neighbors of agent $i$, and likewise $\dot{M}_i \neq {\partial_t} {M}_i = \int_{\mathcal{V}_i} {\partial_t} {\phi} \,\mathrm{d}q$.

A decade later, inspired by \cite{pimenta2009simultaneous}, \cite{lee2015multirobot} proposes the dynamics
\begin{align}
\label{eqn:Lee}
\dot{x} 
= \left(I - \frac{\partial G}{\partial x}\right)^{-1} 
\big(
\kappa (G - x) 
+ 
{\partial_t} G
\big),
\end{align}
where $\kappa >0$, $G = [G_1^\top, \dots, G_N^\top]^\top$ (and the arguments are omitted for notational simplicity). Although the centralized dynamics \eqref{eqn:Lee} ensures asymptotic convergence to a TV-CVT without any quasi-static assumptions on ${\partial_t}\phi$, its decentralized implementation requires truncated Neumann series approximations of the matrix inverse $\left(I - \frac{\partial G}{\partial x}\right)^{-1}$.

From either \eqref{eqn:cortes} or \eqref{eqn:Lee}, one can observe that for each agent~$i$ to compute its dynamics $\dot{x}_i$, it needs to know the dynamics of its Delaunay neighbors, so the corresponding dynamics dependency graph is isomorphic to the underlying Delaunay graph, and the entangled dynamics arises.

More recently, \cite{santos2019decentralized} also explores this problem, but does not achieve decentralized control due to the entangled dynamics, which is discussed in detail in Appendix~\ref{Section:appendix}. Despite the progress in the literature, decentralized coverage control for time-varying density functions without resorting to simplifying approximations (such as quasi-static assumptions or Neumann series truncations) remains a challenging problem that, to the best of our knowledge, has not been resolved yet. 

%% file: sections/disentangled/ctrl.tex
Lemma~\ref{lemma:TVCLF} provides a systematic approach to control synthesis for time-varying objectives and can be directly leveraged by stacking all agents' states and carrying out a centralized synthesis. However, the centralized approach, as discussed in Section~\ref{Section:Introduction}, is not an ideal architecture for multi-agent systems.

For a multi-agent task, if the objective is to minimize $V(x,t) = \sum_{i \in \mathcal{N}} V_i(x_i,t)$ and the cost function for each agent $V_i: \mathbb{R}^{n_i} \times \mathbb{R}_{\geq 0} \to \mathbb{R}_{\geq 0}$ depends only on $x_i$ and $t$, then it can be achieved in a decentralized manner by directly applying Lemma~\ref{lemma:TVCLF} to each agent. However, when the cost function for each agent also depends on the states of other agents or when multiple agents share a single cost function that depends on their states, decentralized control synthesis for each agent becomes nontrivial. This is because the total derivatives of these types of cost functions involve the dynamics of multiple agents simultaneously, which typically requires centralized control synthesis for all the agents involved. 

Against this backdrop, developing a decentralized framework where each agent $i \in \mathcal{N}$ can synthesize its control input $u_i$ without access to any simultaneously computed control input $u_j$, for any $j \in \mathcal{N} \! \setminus \!\!\{i\}$, becomes the focus of our investigation in this section. To this end, we consider two related TVCF structures that lead to entangled dynamics and constitute the main obstacles to directly applying Lemma~\ref{lemma:TVCLF} to multi-agent systems in a decentralized manner. Furthermore, we develop corresponding decentralized control synthesis frameworks to address these two structures, which we refer to as \textbf{\textit{disentangled control}} of multi-agent systems. By \textit{disentangled control} we mean  a control design principle such that the entangled dynamics of multi-agent systems does not occur, thereby enabling decentralization. 
\input{theorem/def1}

Intuitively, the SE-TVCFs arise in scenarios where there exist groupwise shared objectives, such as in multi-agent collaboration. Inspired by the idea in \cite{wang2017safety}, we propose the following optimization framework (in which the arguments are omitted for notational simplicity) for such cost functions shared by a group of agents to synthesize the controller $u_i(\{x_j\}_{j \in \mathcal{N}^{\mathrm{SE}}},t)$ for each agent $i \in \mathcal{N}^{\mathrm{SE}}$ in a decentralized manner.
\input{results/ctrl1}
\input{theorem/thm1}
\input{results/proof2}

For any agent $i \in \mathcal{N}^{\mathrm{SE}}$ to compute its controller $u_i$ by \eqref{eqn:firstOptimization}, it does not require access to $u_j$, $\forall j \in \mathcal{N}^{\mathrm{SE}} \setminus \! \{i\}$, i.e., its dynamics $\dot{x}_i$ does not involve $\dot{x}_j$, $\forall j \in \mathcal{N}^{\mathrm{SE}} \setminus \! \{i\}$. Note that in a specific multi-agent task, there can be multiple different such sets $\mathcal{N}^{\mathrm{SE}}$ in which one agent gets involved. In addition, the priority coefficient $w_i^{\mathcal{N}^{\mathrm{SE}}}$ assigns the relative tightness of each constraint according to its importance or priority, or each agent's characteristics or capabilities, so we also refer to \eqref{eqn:firstOptimization} as \textbf{\textit{constraint allocation}}. Depending on the requirements or preferences of a specific task, such priority coefficients can be fixed, state-dependent, or time-varying, and they can be predefined, adaptively updated, or learned online in practice.

Next, we consider a special form of TVCF associated with each individual agent, which is relevant to the long-standing open problem of decentralized coverage control for time-varying density functions.
\input{theorem/def2}

According to Definition~\ref{def:PCTVCLF}, if the underlying cost graph $\mathcal{G}^{\mathrm{PE}}$ of a multi-agent task has certain special topologies, then Lemma~\ref{lemma:TVCLF} can be implemented in a decentralized manner directly. For example, when $\mathcal{G}^{\mathrm{PE}}$ is a directed path, such as the subgraph enclosed by the blue shaded area in Fig.~\ref{fig:Entangled_graph}, then Lemma~\ref{lemma:TVCLF} can be implemented in a decentralized manner from one agent to the next sequentially. However, our goal is to develop a general framework of decentralized control for multi-agent systems that can be applied to deal with any type of topology of the underlying cost graph $\mathcal{G}^{\mathrm{PE}}$. To this end, we propose the following optimization framework (where the arguments are omitted for notational simplicity), which enables the synthesis of the controller for each agent $i \in \mathcal{N}$ in a decentralized manner under arbitrary topology of $\mathcal{G}^{\mathrm{PE}}$.
\input{results/ctrl2}
\input{theorem/thm2}
\input{results/proof3}

Using the disentangled control framework \eqref{eqn:secondOptimization}, when any agent $i \in \mathcal{N}$ computes its controller $u_i$, it does not require access to $u_j$, $\forall j \in \mathcal{N} \setminus \! \{i\}$, i.e., its dynamics $\dot{x}_i$ does not involve $\dot{x}_j$, $\forall j \in \mathcal{N} \! \setminus \!\! \{i\}$, and the entangled dynamics as illustrated in Fig.~\ref{fig:Entangled_graph} disappears.

Note that although PE-TVCFs in Definition~\ref{def:PCTVCLF} can be viewed as a special case of SE-TVCFs in Definition~\ref{def:SCTVCLF}, their particular structure motivates the specific synthesis method in \eqref{eqn:secondOptimization}. In addition, the conclusions related to time-varying cost functions presented so far also hold for the time-invariant cases. Furthermore, when the optimization problem is infeasible, slack variables, adaptive class $\mathcal{K}_{\infty}$ functions, and/or hierarchical relaxation strategies (see, e.g., \cite{wang2017safety,lee2023hierarchical,lin2025hierarchy,xiao2021adaptive}) may be utilized, which, however, is not the main point of this paper. 

The main advantages of the control synthesis frameworks \eqref{eqn:firstOptimization} and \eqref{eqn:secondOptimization} are as follows.
\input{FigTab/Tab2}

In summary, disentangled control provides a general framework for addressing a wide variety of multi-agent tasks, as it not only handles the problems that can be addressed by existing methods, but also applies to problems beyond their reach, such as the long-standing open problem, decentralized coverage control for time-varying density functions (without Neumann series truncations), which is demonstrated in detail in the following section.

%% file: theorem/def1.tex
\begin{definition}
\label{def:SCTVCLF}
Given a group of agents with index set $\mathcal{N}$, for any $\mathcal{N}^{\mathrm{SE}} \in \mathcal{S}^{\mathrm{SE}} \subseteq 2^{\mathcal{N}} \setminus \!\{\varnothing\}$, where $2^{\mathcal{N}}$ denotes the power set of $\mathcal{N}$, a class $C^1$ function
\begin{equation*}
V_{\mathcal N^{\mathrm{SE}}}:
\mathbb{R}^{\sum_{i\in\mathcal N^{\mathrm{SE}}}n_i}\times\mathbb{R}_{\ge0}
\to\mathbb{R}_{\ge0}
\end{equation*}
is called a shared-entangled TVCF (SE-TVCF) if it depends only on
$\{x_i\}_{i\in\mathcal N^{\mathrm{SE}}}$ and $t$. 
\end{definition}

%% file: results/ctrl1.tex
\vspace{7pt}
\hrule height 0.8pt
\vspace{7pt}
\noindent\textbf{Disentangled Control} (for SE-TVCFs)
\begin{equation}
\label{eqn:firstOptimization}
\begin{aligned}
    & \argmin_{{u}_i \in \mathcal{U}_i}
    & & u_i^\top H_i u_i \\
    & \,\,\,\,\,\,\,\,\textnormal{s.t.}
    & & A_i^{\mathcal{N}^{\mathrm{SE}}} u_i
    + b_i^{\mathcal{N}^{\mathrm{SE}}} \leq 0
\end{aligned}
\end{equation}
in which $H_i \in \mathbb{S}_{++}^{m_i}$,
\begin{equation*}
A_i^{\mathcal{N}^{\mathrm{SE}}} 
= 
{\nabla_{\!x_i}}\!V_{\mathcal{N}^{\mathrm{SE}}} \,g_i,
\end{equation*}
\begin{equation*}
b_i^{\mathcal{N}^{\mathrm{SE}}} 
= 
{\nabla_{\!x_i}}\!V_{\mathcal{N}^{\mathrm{SE}}} \,f_i
+
w_i^{\mathcal{N}^{\mathrm{SE}}} 
\big(  {\partial_t}V_{\mathcal{N}^{\mathrm{SE}}} 
+ \alpha
\left(V_{\mathcal{N}^{\mathrm{SE}}} 
\right)
\big),
\end{equation*}
\begin{equation*}
\sum_{i \in \mathcal{N}^{\mathrm{SE}}} w_i^{\mathcal{N}^{\mathrm{SE}}} = 1, \quad w_i^{\mathcal{N}^{\mathrm{SE}}} \geq 0,
\end{equation*}
and $\alpha:\mathbb{R}_{\ge0}\to\mathbb{R}_{\ge0}$ is of class $\mathcal{K}_{\infty}$.
\vspace{6pt}
\hrule height 0.8pt
\vspace{3pt}

%% file: theorem/thm1.tex
\begin{theorem}
\label{thm:constraintallocation}
For any $\mathcal{N}^{\mathrm{SE}}$ as in {\rm Definition \ref{def:SCTVCLF}}, if every agent $i \in \mathcal{N}^{\mathrm{SE}}$ executes, along a trajectory of \eqref{eqn:controlaffine_ma}, the controller $u_i$ obtained from \eqref{eqn:firstOptimization}, $\forall t \geq 0$, then 
\begin{equation*}
\lim_{t \to \infty}
V_{\mathcal{N}^{\mathrm{SE}}} (\{x_i (t)\}_{i \in \mathcal{N}^{\mathrm{SE}}}, t) = 0.
\end{equation*}
\end{theorem}

%% file: results/proof2.tex
\begin{proof}
Summing $A_i^{\mathcal{N}^{\mathrm{SE}}} u_i + b_i^{\mathcal{N}^{\mathrm{SE}}} 
\leq 0$ over $i \in \mathcal{N}^{\mathrm{SE}}$ yields 
\begin{align*}
&\sum_{i \in \mathcal{N}^{\mathrm{SE}}}
{\nabla_{\!x_i}}\!V_{\mathcal{N}^{\mathrm{SE}}} (\{x_i\}_{i \in \mathcal{N}^{\mathrm{SE}}}, t) \,\dot{x}_i
\notag \\
&+ 
\sum_{i \in \mathcal{N}^{\mathrm{SE}}} w_i^{\mathcal{N}^{\mathrm{SE}}}
\Big(
{\partial_t} V_{\mathcal{N}^{\mathrm{SE}}} (\{x_i\}_{i \in \mathcal{N}^{\mathrm{SE}}}, t) \notag \\
&+ \alpha(V_{\mathcal{N}^{\mathrm{SE}}} (\{x_i\}_{i \in \mathcal{N}^{\mathrm{SE}}}, t))
\Big)
\leq 0
\end{align*}
Since $\sum_{i \in \mathcal{N}^{\mathrm{SE}}} w_i^{\mathcal{N}^{\mathrm{SE}}} = 1$, we can obtain
\begin{equation*}
\frac{d}{dt} V_{\mathcal{N}^{\mathrm{SE}}} (\{x_i\}_{i \in \mathcal{N}^{\mathrm{SE}}},t) 
+
\alpha (V_{\mathcal{N}^{\mathrm{SE}}} (\{x_i\}_{i \in \mathcal{N}^{\mathrm{SE}}},t)) 
\leq 0.
\end{equation*}
By Lemma~\ref{lemma:TVCLF}, this concludes the proof.
\end{proof}

%% file: theorem/def2.tex
\begin{definition}
\label{def:PCTVCLF}
Given a group of agents with index set $\mathcal{N}$, for each $i\in\mathcal{N}$, a class $C^1$ function
\begin{equation*}
V_i^{\mathrm{PE}}: \mathbb{R}^{n_i + \sum_{j \in \mathcal{N}_i^{\mathrm{PE}}} n_j} \times \mathbb{R}_{\geq 0} \to \mathbb{R}_{\geq 0}
\end{equation*}
is called a private-entangled TVCF (PE-TVCF) if it depends only on
$x_i$, $\{x_j\}_{j \in \mathcal{N}_i^{\mathrm{PE}}}$, and $t$, where
$\mathcal{N}_i^{\mathrm{PE}} \subseteq \mathcal{N} \setminus\!\{i\}$ denotes the index set of the agents whose states are involved in $V_i^{\mathrm{PE}}$. The associated cost graph $\mathcal{G}^{\mathrm{PE}}=(\mathcal{N},\mathcal{E}^{\mathrm{PE}})$ contains an edge from $j$ to $i$ whenever $j \in \mathcal{N}_i^{\mathrm{PE}}$, where the edge is undirected if $j \in \mathcal{N}_i^{\mathrm{PE}}$ and $i \in \mathcal{N}_j^{\mathrm{PE}}$.
\end{definition}

%% file: results/ctrl2.tex
\vspace{7pt}
\hrule height 0.8pt
\vspace{7pt}
\noindent\textbf{Disentangled Control} (for PE-TVCFs)
\begin{equation}
\label{eqn:secondOptimization}
\begin{aligned}
    & \argmin_{{u}_i \in \mathcal{U}_i}
    & & u_i^\top H_i u_i \\
    & \,\,\,\,\,\,\,\,\textnormal{s.t.}
    & & A_i^{\mathrm{PE}}\, u_i
    + b_i^{\mathrm{PE}}\leq 0
\end{aligned}
\end{equation}
in which $H_i \in \mathbb{S}_{++}^{m_i}$,
\begin{equation*}
A_i^{\mathrm{PE}} 
= 
{\nabla_{\!x_i}}\!V\,g_i,
\end{equation*}
\begin{equation*}
b_i^{\mathrm{PE}}
= 
{\nabla_{\!x_i}}\!V f_i
+
{\partial_t}V_i^{\mathrm{PE}} 
+ \alpha
\left(V_i^{\mathrm{PE}}
\right),
\end{equation*}
\begin{equation*}
V = \sum_{i \in \mathcal{N}} V_i^{\mathrm{PE}},
\end{equation*}
and $\alpha$ is a subadditive class $\mathcal{K}_{\infty}$ function.
\vspace{7pt}
\hrule height 0.8pt
\vspace{4pt}

%% file: theorem/thm2.tex
\begin{theorem}
\label{theorem:disentanglePE}
If every agent $i \in \mathcal{N}$ executes, along a trajectory of \eqref{eqn:controlaffine_ma}, the controller $u_i$ obtained from \eqref{eqn:secondOptimization}, $\forall t \geq 0$, then
\begin{equation*}
\lim_{t \to \infty} V(x(t),t) 
= 
\lim_{t \to \infty}  \sum_{i \in \mathcal{N}} V_i^{\mathrm{PE}}(x_i(t), \{x_j(t)\}_{j \in \mathcal{N}_i^{\mathrm{PE}}}, t) = 0. 
\end{equation*}
\end{theorem}

%% file: results/proof3.tex
\begin{proof}
Summing $A_i^{\mathrm{PE}} u_i + b_i^{\mathrm{PE}} \leq 0$ over $i \in \mathcal{N}$ yields
\begin{align*}
&{\nabla_{\!x}}V(x,t)
\,\dot{x} 
+ 
{\partial_t} V(x,t)
\leq 
- \sum_{i \in \mathcal{N}}
\alpha
(
V_i^{\mathrm{PE}}(x_i, \{x_j\}_{j \in \mathcal{N}_i^{\mathrm{PE}}}, t)
).
\end{align*}
Since $\alpha$ is a subadditive class $\mathcal{K}_{\infty}$ function, we have
\begin{equation*}
\sum_{i \in \mathcal{N}}
\alpha
(
V_i^{\mathrm{PE}}(x_i, \{x_j\}_{j \in \mathcal{N}_i^{\mathrm{PE}}}, t)
)
\geq 
\alpha (V(x,t)).
\end{equation*}
Thus, we can obtain
\begin{equation*}
\frac{d}{dt} V(x,t) +
\alpha(V(x,t))
\leq 0.
\end{equation*}
Therefore, using Lemma~\ref{lemma:TVCLF}, the proof is completed.
\end{proof}

%% file: FigTab/Tab2.tex
\begin{itemize}
\item \input{sections/adv1}
\item \input{sections/adv2}
\item \input{sections/adv3}
\item \input{sections/adv4}
\end{itemize}

%% file: sections/adv1.tex
They are decentralized in the sense that when agent~$i$, for any $i \in \mathcal{N}$, computes its control input $u_i$, it does not require access to any simultaneously computed control input $u_j$, for any $j \in \mathcal{N} \setminus \! \{i\} $.

%% file: sections/adv2.tex
The disentangled control framework \eqref{eqn:secondOptimization} is applicable to cost graphs of arbitrary topology, including directed, undirected, and mixed graphs.

%% file: sections/adv3.tex
They are quadratic programs under affine control input constraints, and thus are ideal for real-time applications.

%% file: sections/adv4.tex
They accommodate simultaneous consideration of multiple objectives, as any additional constraint that encodes a certain objective can be plugged into them.

%% file: sections/AppExp/app.tex
In this section, we apply the results developed in Section~\ref{Section:DisentangledControl} with $H_i = I$, $f_i = 0$, $g_i = I$, $\forall i \in \mathcal{N}$, to a team of mobile robots, a common object in multi-agent studies, for different tasks, where the state of interest refers to the positions of the robots (with $n_i=2$ for planar robots and $n_i=3$ for aerial robots). For physical experiments, we utilize a group of differential-drive mobile robots in a rectangular domain $\mathcal{D} \subset \mathbb{R}^2$ with $x$-axis ranging from $-1.6$ to $1.6$ meters and $y$-axis ranging from $-1$ to $1$ meter \cite{wilson2020robotarium}. Each robot operates with the help of the near-identity diffeomorphism \cite{olfati2002near}. Note that all the applications demonstrated in this section cannot be theoretically addressed with rigorous convergence guarantees under the methods proposed in \cite{notomista2019constraint,santos2019decentralized}, which is explained in detail in Appendix~\ref{Section:appendix}.

\subsection{Time-Varying Leader-Follower Formation Control} \label{Section:subsec_exp1}
\input{sections/AppExp/formation}

\subsection{Decentralized Coverage Control for Time-Varying Density Functions} \label{Section:subsec_exp2}
\input{sections/AppExp/coverage}

\subsection{Safe Formation Navigation in a Dense Environment} \label{Section:subsec_exp3}
\input{sections/AppExp/navigation}

%% file: sections/AppExp/formation.tex
\input{FigTab/Fig2}
\input{FigTab/Fig3}
In this subsection, we employ disentangled control to address time-varying leader-follower formation control, which can be applied to collaborative transportation in space, object caging or containment, swarm robotic art performance, etc.

Specifically, we consider $N=13$ robots with $12$ followers, whose states are denoted by $x_i \in \mathbb{R}^3$, $\forall i \in \mathcal{N} \!\setminus\!\! \{\ell\}$, and one leader, whose state trajectory is prescribed by $x_{\ell}(t) \in \mathbb{R}^3$ with $\ell = 13$. The goal is that the follower robots form a regular icosahedron of time-varying size, containing the leader robot in its geometric center. To this end, the PE-TVCF for each follower robot is defined as
\begin{align}
V_i^{\mathrm{PE}}(x_i, \{x_j\}_{j \in \mathcal{N}_i^{\mathrm{PE}}}, t)
&= 
\frac{w^{\mathrm{f}}}{2} \!\!\!\!
\sum_{j \in \mathcal{N}_i^{\mathrm{PE}}\setminus \!\{\ell\}}
\!\!\!\!\!
\left(
\| x_i - x_j\| - d_{ij}(t)
\right)^2 \notag \\
&\,\,\, + 
\frac{w^{\mathrm{\ell}}}{2}\!
\left(
\| x_i - x_{\ell}(t)\| - d_{i}^\ell(t)
\right)^2,
\label{eqn:VPE_i_exp1}
\end{align}
$\forall i \in \mathcal{N}\setminus\! \{\ell\}$, where the cost graph $\mathcal{G}^{\mathrm{PE}}$ consists of the vertices corresponding to the follower robots and edges of a regular icosahedron, augmented with an additional vertex corresponding to the leader robot and directed edges from the leader to each follower, $d_{ij}(t)>0$ denotes the edge length at time $t$, $d_i^\ell(t)>0$ denotes the circumradius (i.e., the distance from the geometric center to a vertex) at time $t$, and the weighting coefficients $w^{\mathrm{f}} >0$ and $w^{\mathrm{\ell}} >0$ determine the relative importance of achieving the goals between maintaining the time-varying formation and following the leader. 
\input{FigTab/Fig4}

According to \eqref{eqn:VPE_i_exp1}, we can obtain
\begin{align}
{{\nabla_{\!x_i}}\!V} 
&= 
2 w^{\mathrm{f}} \!\!\!
\sum_{j \in \mathcal{N}_i^{\mathrm{PE}}\setminus \!\{\ell\}}\! 
\frac{\|x_i - x_j\| - d_{ij}(t)}{\|x_i - x_j\|} (x_i - x_j)^\top \notag \\
&\quad\, +
w^{\mathrm{\ell}} \frac{\|x_i - x_{\ell}(t)\| - d_i^{\ell}(t)}{\|x_i - x_{\ell}(t)\|} (x_i - x_{\ell}(t))^\top
\label{eqn:gradVPE_i_exp1}
\end{align}
and
\begin{align}
&{\partial_t} V_i^{\mathrm{PE}}
=
w^{\mathrm{f}} \!\!\!
\sum_{j \in \mathcal{N}_i^{\mathrm{PE}}\setminus\!\{\ell\}}\! 
(d_{ij}(t) - \|x_i - x_j\|) \, \dot{d}_{ij}(t) \notag \\
& +
w^{\mathrm{\ell}} (d_i^\ell(t) - \|x_i - x_\ell(t)\|)
\left(
\frac{(x_i - x_\ell(t))^\top}{\|x_i - x_\ell(t)\|}\dot{x}_\ell(t) + \dot{d}_i^\ell(t)
\right),
\label{eqn:pVPEpt_i_exp1}
\end{align}
$\forall i \in \mathcal{N}\!\setminus\!\! \{\ell\}$. Substituting \eqref{eqn:VPE_i_exp1}, \eqref{eqn:gradVPE_i_exp1}, and \eqref{eqn:pVPEpt_i_exp1} into the disentangled control framework \eqref{eqn:secondOptimization}, the simulation results are shown in Fig.~\ref{fig:exp1}, Fig.~\ref{fig:exp1_curve}, and Fig.~\ref{fig:exp1_curve_sub}. The results show that the total PE-TVCF of the follower robots $V^{\mathrm{PE}}(x,t) = \sum_{i \in \mathcal{N}\setminus \!\{\ell\}} V_i^{\mathrm{PE}}(x_i, \{x_j\}_{j \in \mathcal{N}_i^{\mathrm{PE}}}, t)$ decreases to $0$ at around time step $300$ and remains $0$ thereafter (within numerical precision), indicating that the twelve follower robots closely track the desired regular icosahedral formation with time-varying size and contain the leader robot near the center of the icosahedron. Note that the externally commanded leader trajectory $x_{\ell}(t)$ is prescribed independently of the follower states and is treated as an explicit time-dependent reference. It can be seen from \eqref{eqn:VPE_i_exp1}, \eqref{eqn:gradVPE_i_exp1}, and \eqref{eqn:pVPEpt_i_exp1} that the problem of time-varying leader-follower formation control can be solved using \eqref{eqn:secondOptimization} in a decentralized manner because each follower robot computes its control input without access to the simultaneously computed control input of any other follower robot.

Notably, we observe from Fig.~\ref{fig:exp1_curve_sub} (together with Fig.~\ref{fig:exp1_curve}) that an interesting phenomenon emerges spontaneously (i.e., without the imposition of certain predefined rules aimed at producing it): To ensure the decrease of the total PE-TVCF $V^{\mathrm{PE}} = \sum_{i \in \mathcal{N}\setminus \!\{\ell\}} V_i^{\mathrm{PE}}$, some follower robots occasionally ``sacrifice'' by allowing their individual PE-TVCFs $V_i^{\mathrm{PE}}$ to increase during some period of time (e.g., robot~$9$ from around time steps $15$ to $35$ and $100$ to $150$), thereby contributing to the ``greater good'' of the whole system.

%% file: FigTab/Fig2.tex
\begin{figure*}[t]
\centering
\subfigure[Time Step $1$]{
\includegraphics[width=0.23\textwidth]{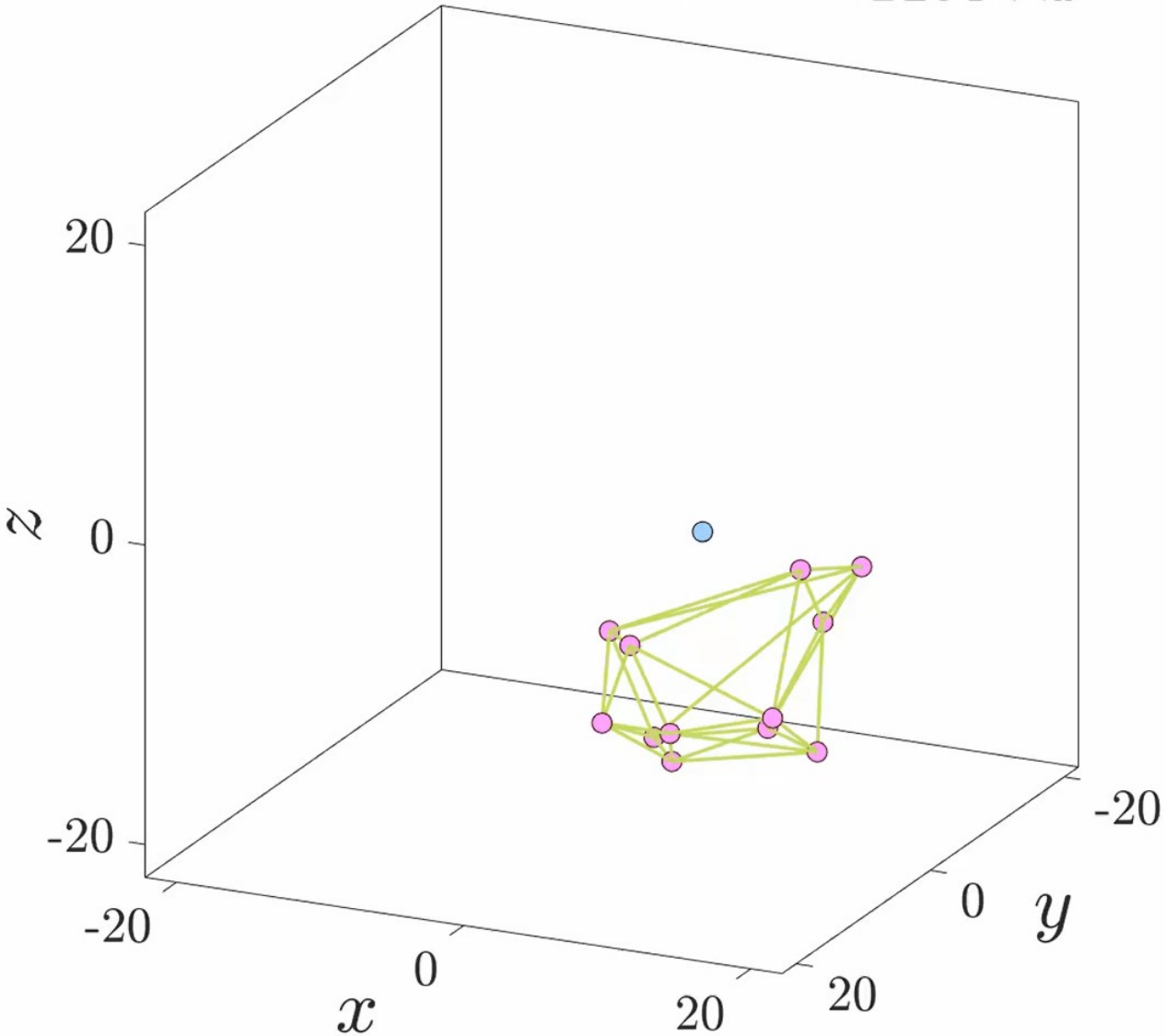}
\label{fig:exp1_1}
}
\subfigure[Time Step $300$]{
\includegraphics[width=0.23\textwidth]{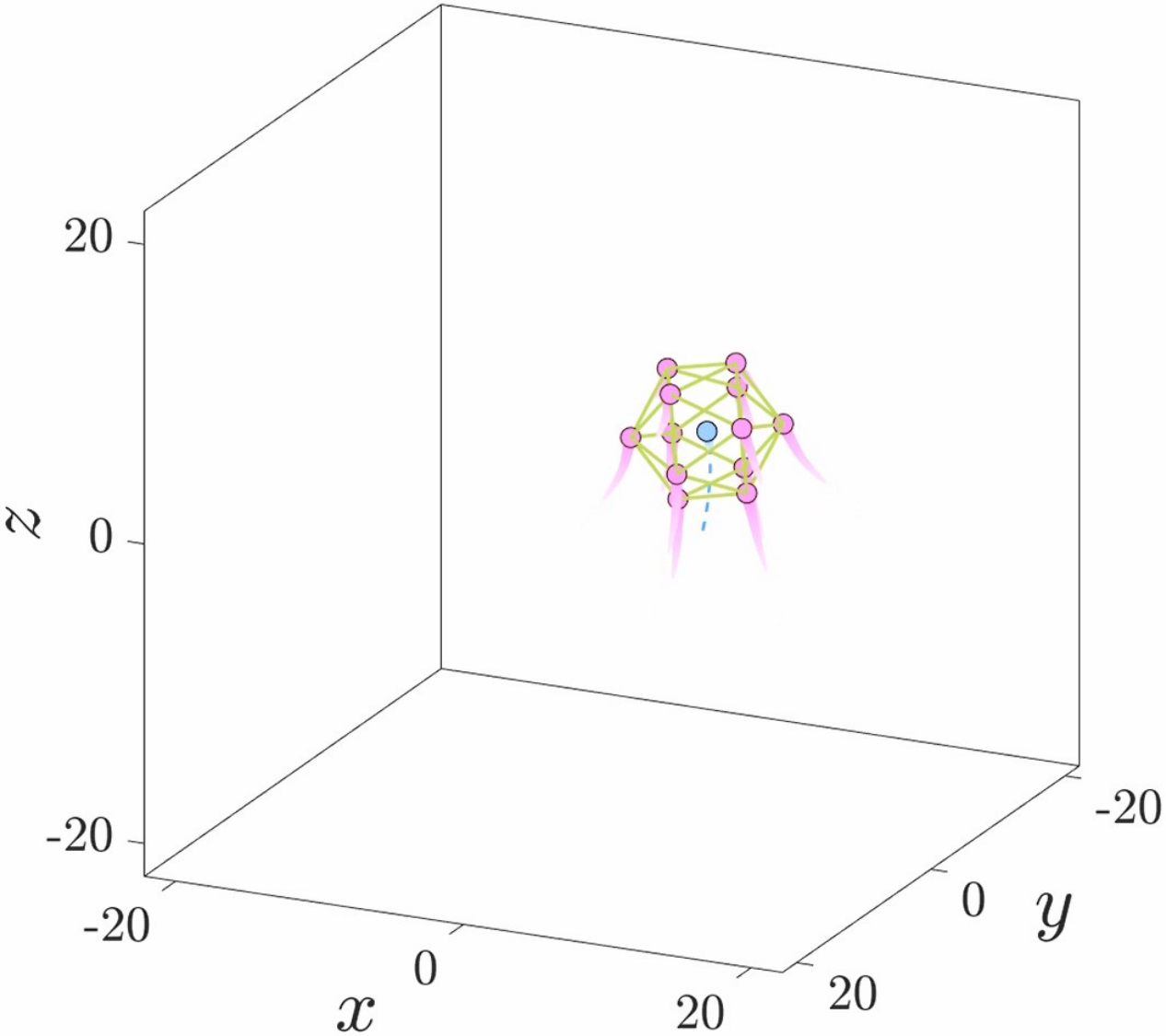}
\label{fig:exp1_300}
}
\subfigure[Time Step $6000$]{
\includegraphics[width=0.23\textwidth]{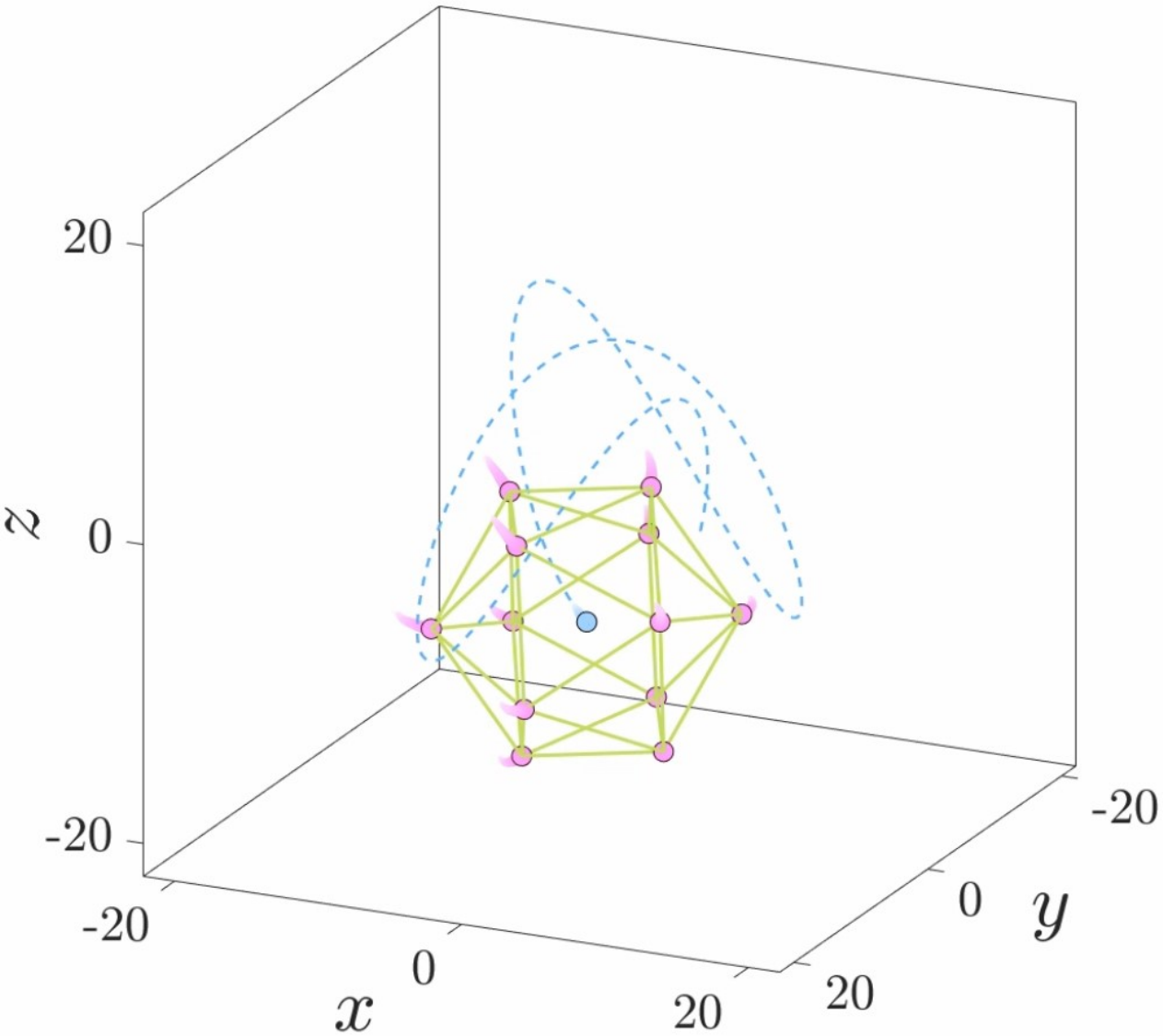}
}
\subfigure[Time Step $7000$]{
\includegraphics[width=0.23\textwidth]{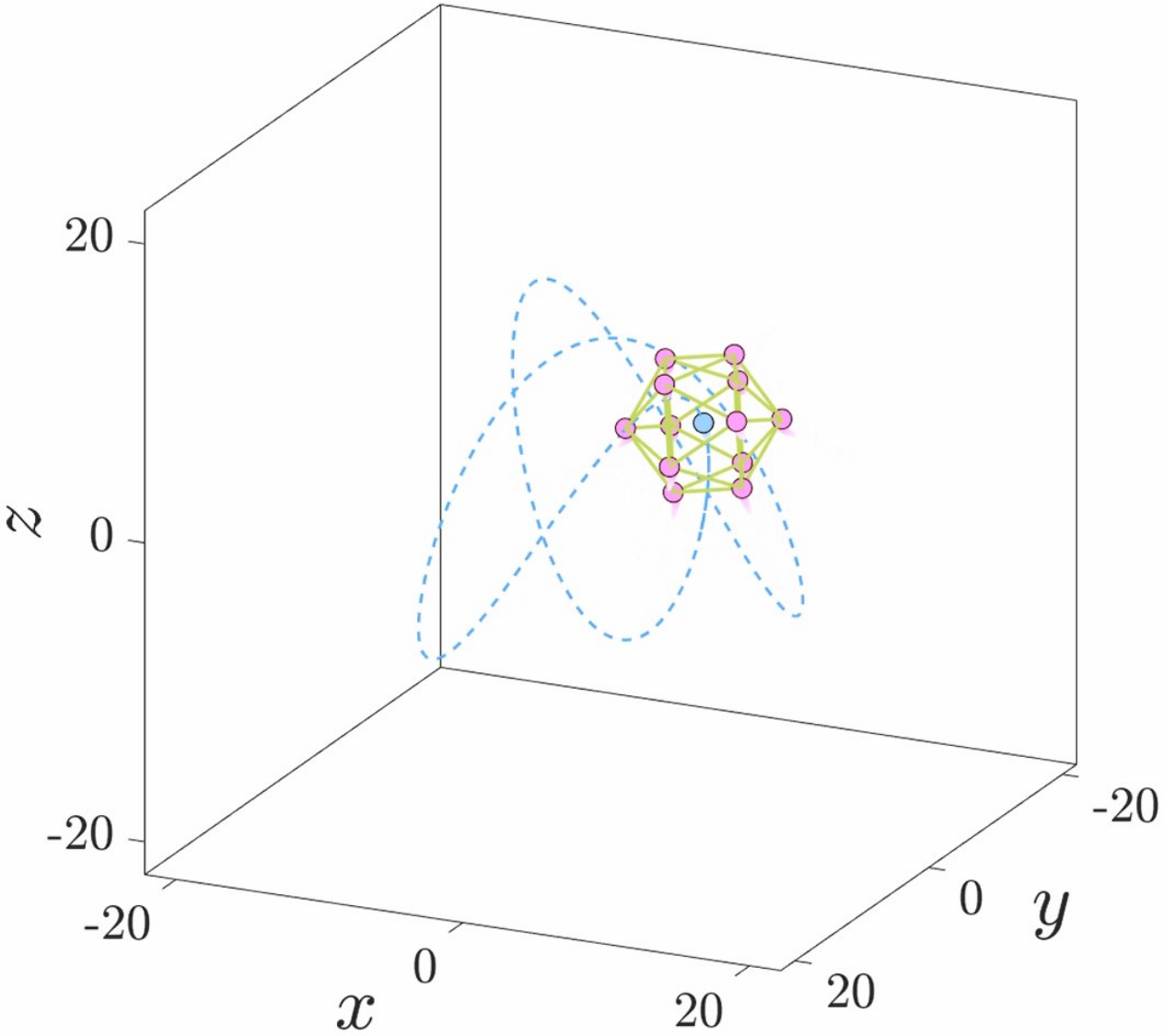}
}
\caption{Snapshots of the simulation of decentralized time-varying leader-follower formation control in Section~\ref{Section:subsec_exp1}. The twelve follower robots are represented by the pink dots, while the leader robot is represented by the blue dot. The green lines represent the undirected edges of the corresponding cost graph, and the blue dashed line represents the trajectory traveled by the leader robot. The full video of this experiment is available online at \url{https://youtu.be/DyY8JiIXzAc}.
}
\label{fig:exp1}
\end{figure*}

%% file: FigTab/Fig3.tex
\begin{figure}[t]
\centering
\includegraphics[scale=0.3]{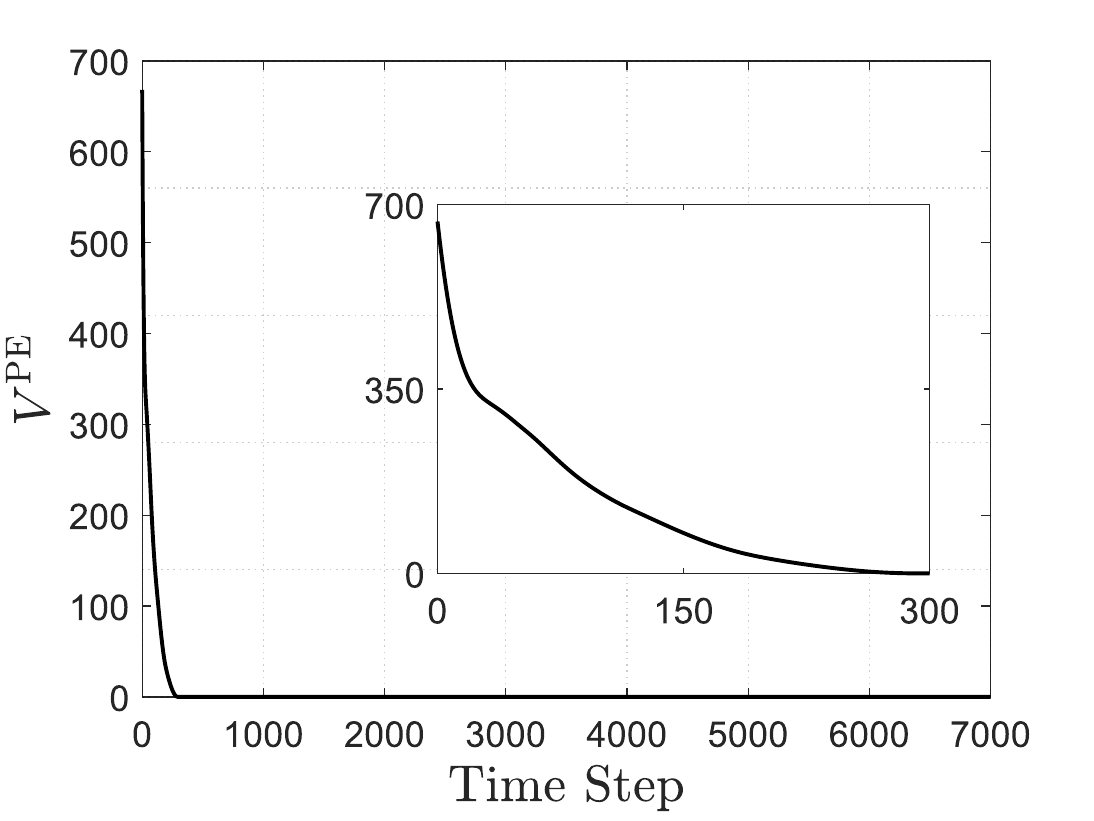}
\caption{The evolution of the total PE-TVCF of the follower robots defined in Section~\ref{Section:subsec_exp1} with respect to time steps.}
\label{fig:exp1_curve}
\end{figure}

%% file: FigTab/Fig4.tex
\begin{figure}[t]
\centering
\includegraphics[scale=0.3]{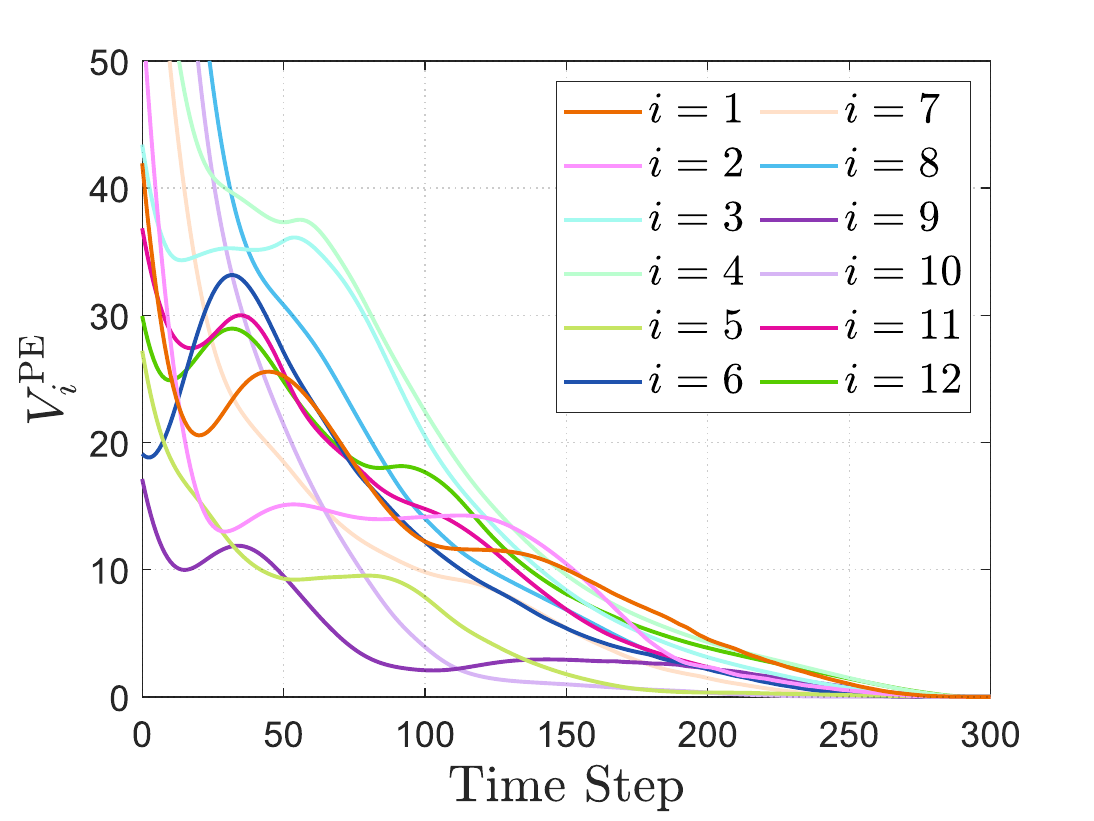}
\caption{The evolution of the PE-TVCF for each follower robot defined in Section~\ref{Section:subsec_exp1} with respect to time steps.}
\label{fig:exp1_curve_sub}
\end{figure}

%% file: sections/AppExp/coverage.tex
\input{FigTab/Fig5}
\input{FigTab/Fig6}

In this subsection, we employ disentangled control to address the long-standing open problem, decentralized coverage control for time-varying density functions (without Neumann series approximations), which can be applied to environmental monitoring, target tracking, search-and-rescue, etc.

Since asymptotic convergence to a TV-CVT is the coverage objective, the PE-TVCF for each robot can be formulated as
\begin{align}
V_i^{\mathrm{PE}}(x_i, \{x_j\}_{j \in \mathcal{N}_i^{\mathrm{PE}}}, t)
= \frac{1}{2}\|x_i - G_i(x_i, \{x_j\}_{j \in \mathcal{N}_i^{\mathrm{PE}}}, t)\|^2,
\label{eqn:VPE_i_exp2}
\end{align}
$\forall i \in \mathcal{N}$, where the cost graph $\mathcal{G}^{\mathrm{PE}}$ is isomorphic to the Delaunay graph. Precisely, 
\begin{equation}
G_i(x_i, \{x_j\}_{j \in \mathcal{N}_i^{\mathrm{PE}}}, t) = \frac{\int_{\mathcal{V}_i(x_i,\{x_j\}_{j \in \mathcal{N}_i^{\mathrm{PE}}})}\! q \phi(q,t)
\,\mathrm{d}q}{M_i(x_i, \{x_j\}_{j \in \mathcal{N}_i^{\mathrm{PE}}},t)}, \notag
\end{equation}
in which 
\begin{equation}
M_i(x_i, \{x_j\}_{j \in \mathcal{N}_i^{\mathrm{PE}}},t) = \int_{\mathcal{V}_i(x_i, \{x_j\}_{j \in \mathcal{N}_i^{\mathrm{PE}}})}\! \phi(q,t)
\,\mathrm{d}q. \notag
\end{equation}

According to \eqref{eqn:VPE_i_exp2}, we can obtain
\begin{equation}
{\partial_t} V_i^{\mathrm{PE}}
= 
\frac{1}{M_i}(G_i - x_i)^\top \!
\int_{\mathcal{V}_i}\! (q - G_i) \, 
{\partial_t} \phi
\,\mathrm{d}q,
\label{eqn:pVPEpt_i_exp2}
\end{equation}
and
\begin{equation}
\label{eqn:gradVPE_i_exp2}
{{\nabla_{\!x_i}}\!V} 
= 
(x_i - G_i)^\top \!
\left(
\!I - \frac{\partial G_i}{\partial x_i}\!
\right)
+ \!\!
\sum_{j \in \mathcal{N}_i^{\mathrm{PE}}}\!\! 
\left(
G_j - x_j
\right)^\top \frac{\partial G_j}{\partial x_i},
\end{equation}
in which the Jacobians are
\begin{equation}
\frac{\partial G_i}{\partial x_i}
= \!\!\!
\sum_{j \in \mathcal{N}_i^{\mathrm{PE}}} \!\!
\frac{
\int_{\partial \mathcal{V}_{ij}} (q - G_i) (q - x_i)^{\top} \phi
\,\mathrm{d}q
}{
M_i \, \|x_j - x_i\|
}, \notag
\end{equation}
\begin{equation}
\frac{\partial G_j}{\partial x_i}
=
\frac{
\int_{\partial \mathcal{V}_{ij}} (G_j - q)  (q - x_i)^{\top}  \phi \, \mathrm{d}q
}{
M_j \, \|x_i - x_j\|
}, \notag
\end{equation}
where $\partial \mathcal{V}_{ij}$ denotes the shared boundary between the Voronoi cells $\mathcal{V}_i$ and $\mathcal{V}_j$, $\mathrm{d}q$ denotes the corresponding measure, and the arguments in \eqref{eqn:pVPEpt_i_exp2} and \eqref{eqn:gradVPE_i_exp2} are omitted for notational simplicity.

It can be seen from \eqref{eqn:VPE_i_exp2}, \eqref{eqn:pVPEpt_i_exp2}, and \eqref{eqn:gradVPE_i_exp2} that the problem of time-varying coverage control can be solved using \eqref{eqn:secondOptimization} in a decentralized manner since each robot computes its control input without access to the simultaneously computed control input of any other robot. Moreover, no simplifying approximation, such as Neumann series truncation, is involved, and asymptotic convergence to a TV-CVT,
\begin{equation}
x_i^{*}(t) = G_i(x_i^{*}(t),\{x_j^{*}(t)\}_{j \in \mathcal{N}_i^{\mathrm{PE}}}, t), \quad \forall i \in \mathcal{N},
\label{eqn:TVCVT}
\end{equation}
follows from Theorem~\ref{theorem:disentanglePE}. Note that although the underlying Delaunay graph can be piecewise-constant, $V = \sum_{i \in \mathcal{N}}V_i^{\mathrm{PE}}$ is continuous across graph switches and \eqref{eqn:secondOptimization} ensures that $V$ decreases on each interval where the Delaunay graph remains constant.

In the physical experiment, we use a group of $N=10$ differential-drive mobile robots. The initial configuration of the robots is shown in Fig.~\ref{fig:exp2_1}, where $\mathcal{V}_i$ and $G_i$, for any $i \in \mathcal{N}$, are represented by a set of white lines and a white dot, respectively, and the density $\phi(q,t)$ is visualized by a rainbow colormap with warmer colors corresponding to higher density values and cooler colors to lower values. Substituting \eqref{eqn:VPE_i_exp2}, \eqref{eqn:pVPEpt_i_exp2}, and \eqref{eqn:gradVPE_i_exp2} into \eqref{eqn:secondOptimization}, the experimental results are shown in Fig.~\ref{fig:exp2} and Fig.~\ref{fig:exp2_curve}. As shown in Fig.~\ref{fig:exp2_curve}, the total PE-TVCF $V(x,t) = \sum_{i \in \mathcal{N}} V_i^{\mathrm{PE}}(x_i, \{x_j\}_{j \in \mathcal{N}_i^{\mathrm{PE}}}, t)$ decreases to $0$ at around time step $140$ and remains $0$ thereafter (within numerical precision), indicating that the team of robots closely tracks a TV-CVT \eqref{eqn:TVCVT}, as can also be seen in Fig.~\ref{fig:exp2}.

%% file: FigTab/Fig5.tex
\begin{figure*}[t]
\centering
\subfigure[Time Step $1$]{
\includegraphics[width=0.23\textwidth]{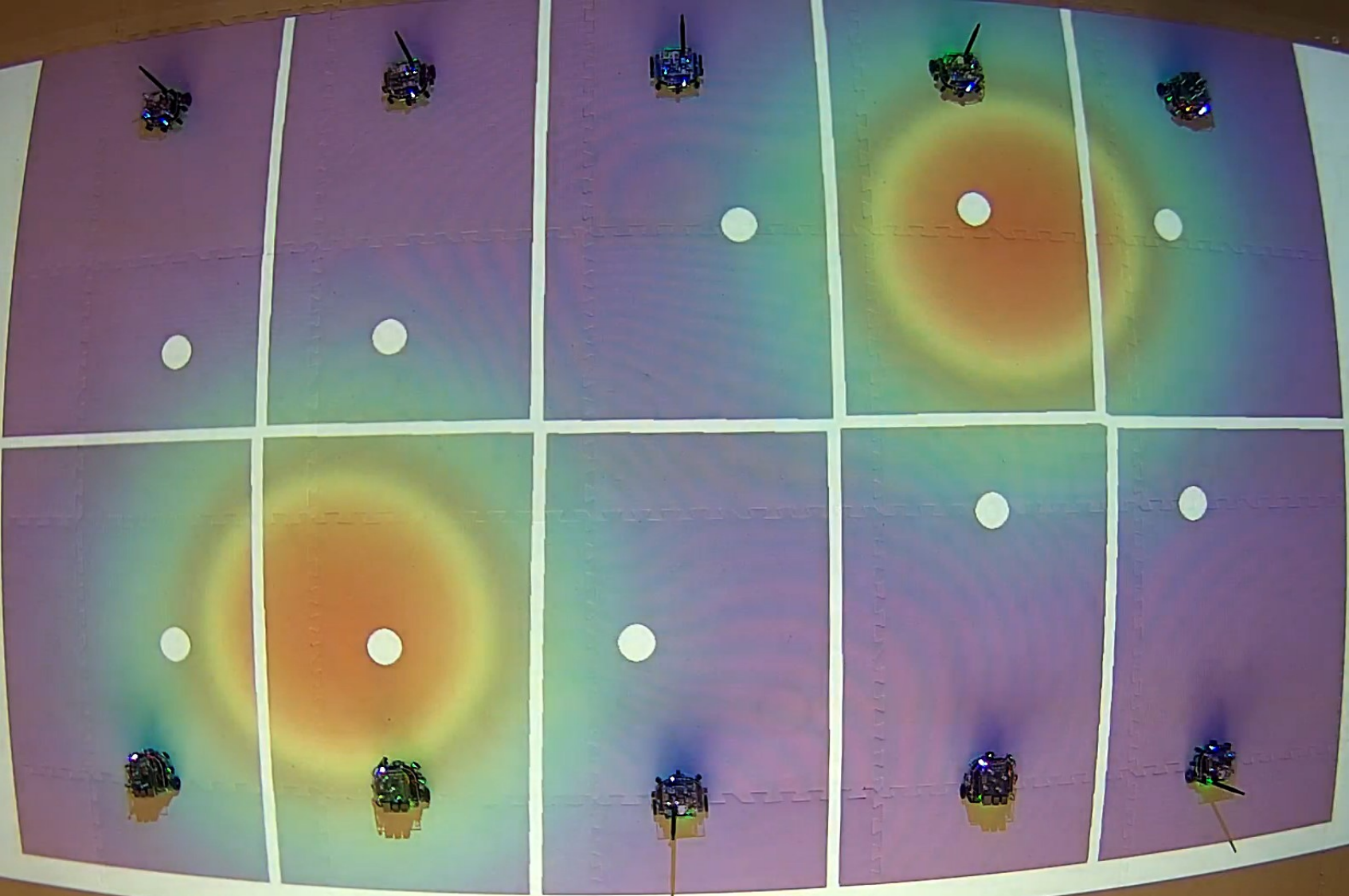}
\label{fig:exp2_1}
}
\subfigure[Time Step $300$]{
\includegraphics[width=0.23\textwidth]{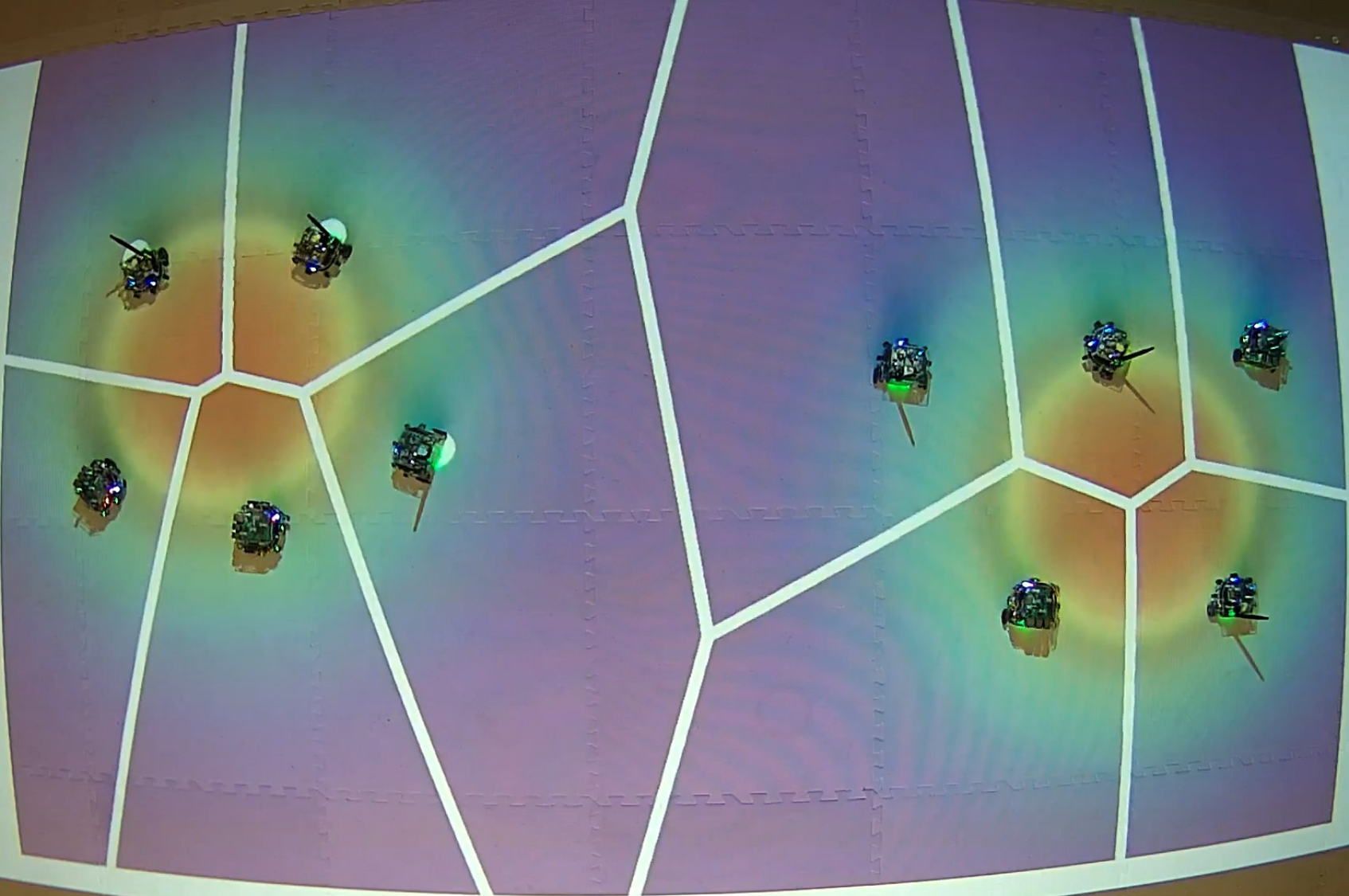}
}
\subfigure[Time Step $800$]{
\includegraphics[width=0.23\textwidth]{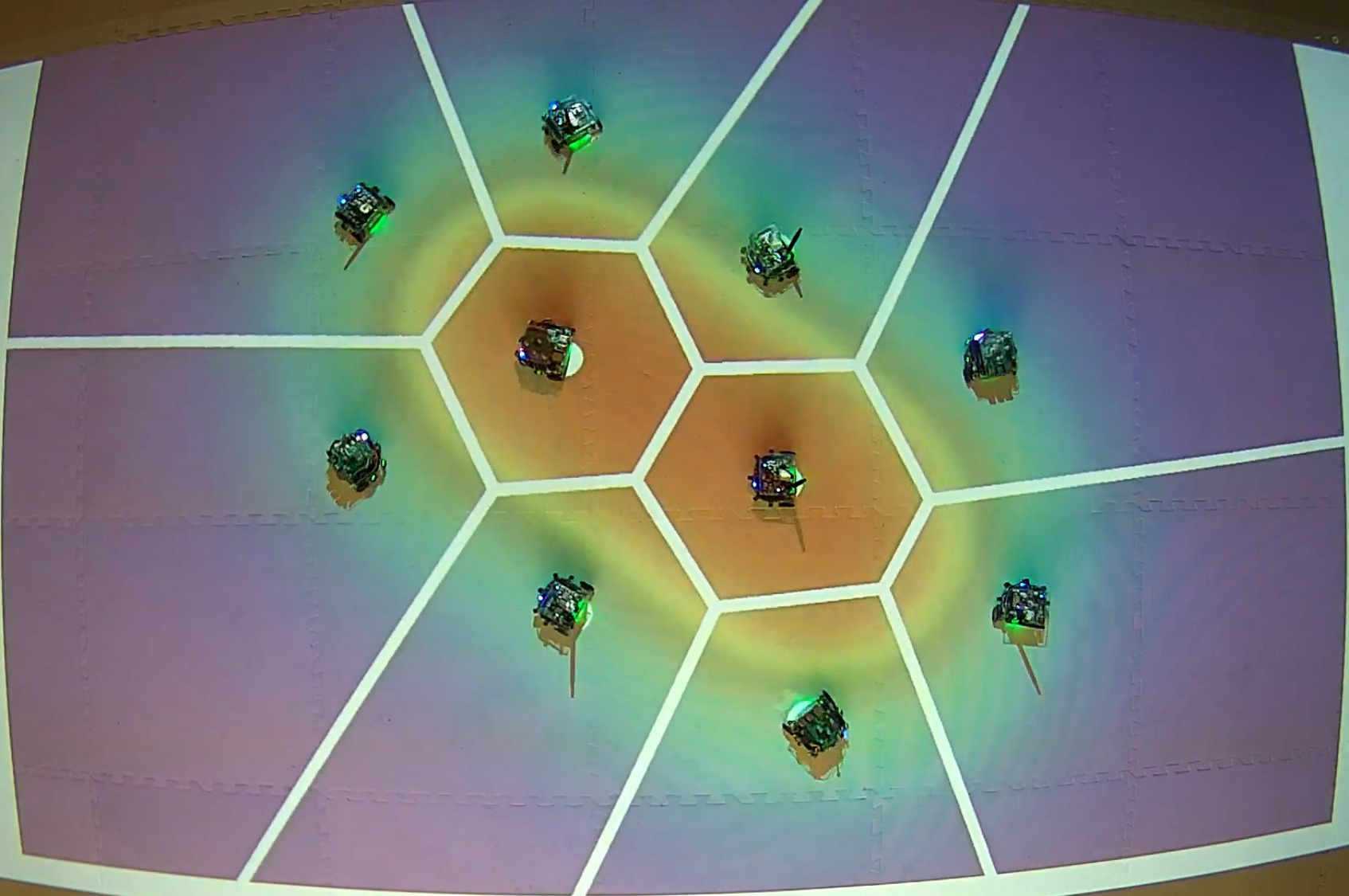}
}
\subfigure[Time Step $1400$]{
\includegraphics[width=0.23\textwidth]{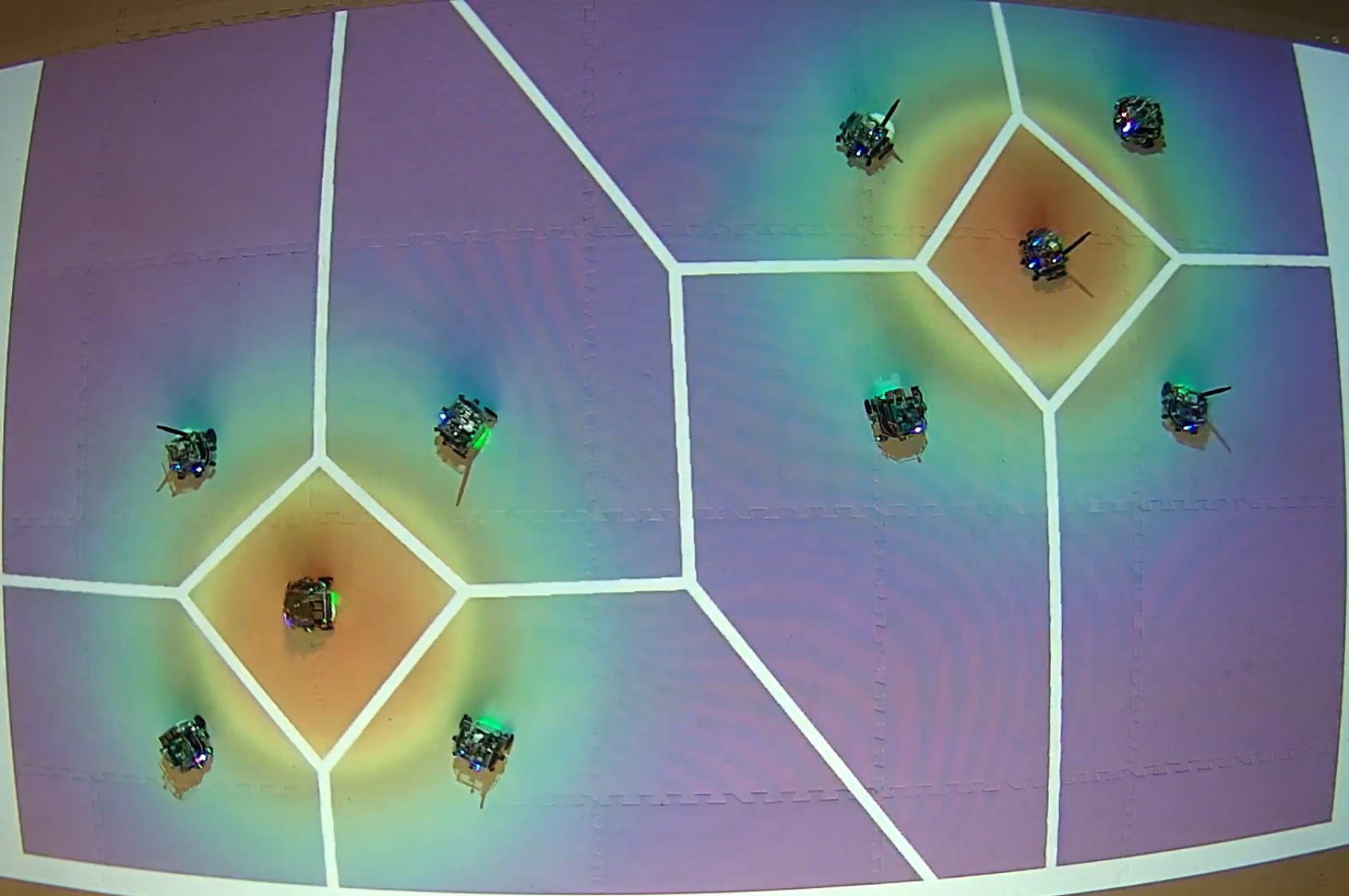}
}
\caption{Snapshots of the physical experiment of decentralized coverage control for time-varying density distribution in Section~\ref{Section:subsec_exp2}. Each robot is in charge of a dominant subdomain, i.e., a Voronoi cell, with a set of white lines as its boundaries and a white dot as its center of mass with respect to the time-varying density function. A rainbow colormap is employed to visualize this density distribution, with warmer colors representing higher density values and cooler colors representing lower values. The full video of this experiment is available online at \url{https://youtu.be/eyFkZO9CuQ0}.
}
\label{fig:exp2}
\end{figure*}

%% file: FigTab/Fig6.tex
\begin{figure}[t]
\centering
\includegraphics[scale=0.3]{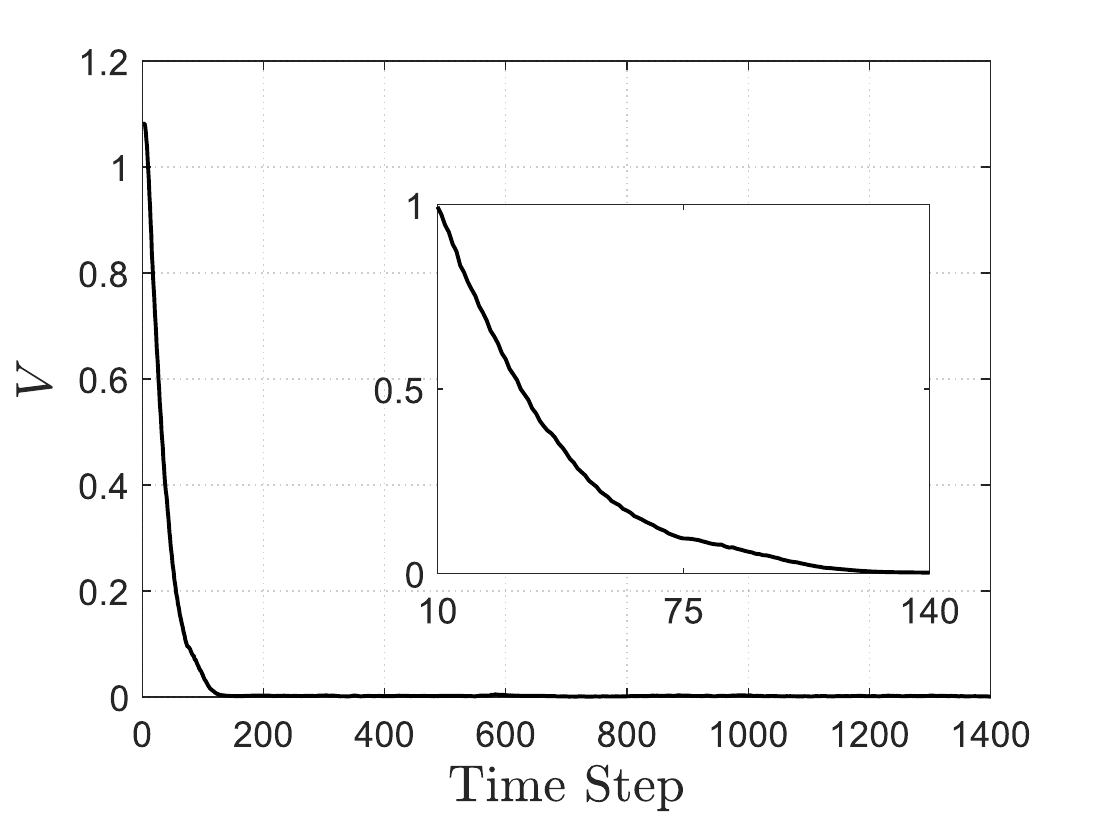}
\caption{The evolution of the total PE-TVCF defined in Section~\ref{Section:subsec_exp2} with respect to time steps.}
\label{fig:exp2_curve}
\end{figure}

%% file: sections/AppExp/navigation.tex
\input{FigTab/Fig7}
\input{FigTab/Fig8}

As mentioned in Section~\ref{Section:subsecdisentangledcontrol}, an advantage of disentangled control is that it naturally accommodates multi-objective robotic tasks through incorporating the corresponding constraints. To demonstrate this benefit, in this subsection, we consider the objective of collision avoidance in addition to navigation in formation.

Specifically, we consider $N=5$ agents with $4$ follower robots, whose states are denoted by $x_i \in \mathbb{R}^2$, $\forall i \in \mathcal{N} \!\setminus \!\!\{\ell\}$, and one virtual leader, whose state is prescribed by $x_{\ell}(t) \in \mathbb{R}^2$ with $\ell = 5$. The goal is that the follower robots form a square and follow the virtual leader as much as possible while avoiding collisions in an environment with densely distributed circular obstacles. To this end, the PE-TVCF for each follower robot is defined as
\begin{align}
V_i^{\mathrm{PE}}(x_i, \{x_j\}_{j \in \mathcal{N}_i^{\mathrm{PE}}},t)
&= 
\frac{w^{\mathrm{f}}}{2} \!\!\!
\sum_{j \in \mathcal{N}_i^{\mathrm{PE}}\setminus \!\{\ell\}}
\!\!\!
\left(
\| x_i - x_j\| - d_{ij}
\right)^2 \notag \\
&\quad+
\frac{w^{\mathrm{\ell}}}{2}
\left(
\| x_i - x_{\ell}(t)\| - d_{i}^\ell
\right)^2,
\label{eqn:VPE_i_exp3}
\end{align}
which is a special case of \eqref{eqn:VPE_i_exp1}, $\forall i \in \mathcal{N}\setminus\! \{\ell\}$, where the cost graph $\mathcal{G}^{\mathrm{PE}}$ consists of the vertices representing the follower robots and edges of a square, augmented with an additional vertex representing the virtual leader and directed edges from the leader to each follower. 

Since each robot has to avoid colliding with obstacles and safety is typically more important than task accomplishment, a slack variable should be added into the constraint induced by \eqref{eqn:VPE_i_exp3}. In addition, due to the slack variable introduced into the constraint induced by \eqref{eqn:VPE_i_exp3}, collision avoidance between robots should also be taken into account. Hence, we set two additional types of constraints for each robot to ensure that it neither collides with obstacles nor with any other robots. 

The control barrier function (CBF) (see \cite{ames2016control} for detailed discussion on CBF) for robot $i$ to avoid colliding with obstacle $k \in \mathcal{K} \coloneqq \{1,2,\dots,K\}$ with $K \in \mathbb{Z}^+$ is given by $h_i^k(x_i) = \|x_i - c_k\|^2 - r_k^2$, where $c_k \in \mathbb{R}^2$ and $r_k >0$ are the center and radius of obstacle $k$ respectively. In addition, the CBF for robot $i$ to avoid colliding with robot $j$ is given by \cite{wang2017safety} 
\begin{equation}
\label{eqn:CAcbf}
h_{ij}^{\mathrm{SE}}(x_i,x_j) = \|x_i - x_j\|^2 - D_{ij}^2,
\end{equation}
where $D_{ij}>0$ is the safe distance between robots $i$ and $j$.
Note that \eqref{eqn:CAcbf} is a shared-entangled CBF (SE-CBF) that involves the states of both robots $i$ and $j$, as defined below.
\begin{definition}
Given a group of agents with index set $\mathcal{N}$, for any $\mathcal{N}^{\mathrm{SE}} \in \mathcal{S}^{\mathrm{SE}} \subseteq 2^{\mathcal{N}} \setminus \!\{\varnothing\}$, a CBF $h_{\mathcal{N}^{\mathrm{SE}}}: \mathbb{R}^{\sum_{i \in \mathcal{N}^{\mathrm{SE}}}n_i}
\to \mathbb{R}$
is called an SE-CBF if it depends only on $\{x_i\}_{i \in \mathcal{N}^{\mathrm{SE}}}$.
\end{definition}

Thus, disentangled control with safety requirements considered in this subsection takes the form of
\begin{align}
    & \argmin_{{u}_i \in \mathcal{U}_i, \,\delta_i \in \mathbb{R}} \!\!\!\!\!\!
    & & \|u_i\|^2 + \kappa_i \delta_i^2 \notag \\
    & \,\,\,\,\,\,\,\,\,\,\,\,\textnormal{s.t.}
    & & 
    {{\nabla_{\!x_i}}\!V} 
    \,u_i \leq - \alpha\!\left(V_i^{\mathrm{PE}}\right) - {\partial_t} V_i^{\mathrm{PE}} + \delta_i \notag \\
    & & & 
    (x_i - x_j)^\top u_i \geq -w_i^{\mathcal{N}^{\mathrm{SE}}} \gamma (h_{ij}^{\mathrm{SE}}), \;\; \forall j \in \mathcal{N} \!\setminus\!\! \{i,\ell\}, \notag \\
    & & & 
    (x_i - c_k)^\top u_i \geq -\gamma (h_i^k), 
    \;\; \forall k \in \mathcal{K},
    \notag
\end{align}
$\forall i \in \mathcal{N}\setminus\! \{\ell\}$, where $\kappa_i >0$ is a penalty coefficient, $\gamma: \mathbb{R} \to \mathbb{R}$ is an extended class $\mathcal{K}_{\infty}$ function, and an additional constraint with a slack variable is added when $i=2$ to regulate the orientation of the formation.

The initial configuration of four differential-drive mobile robots is shown in Fig.~\ref{fig:exp3_1}. The experimental results are shown in Fig.~\ref{fig:exp3} and Fig.~\ref{fig:exp3_curve}, from which one can observe that at the beginning and from around time step $1600$ to $2000$, the four robots are in the desired square formation, i.e., the corresponding total PE-TVCF of the follower robots $V^{\mathrm{PE}}= \sum_{i \in \mathcal{N}\setminus \!\{\ell\}} V_i^{\mathrm{PE}}$ is zero. However, when the robots are navigating through the densely distributed obstacles, due to the priority of safety (i.e., collision avoidance), they cannot maintain a perfect formation, which is why the total PE-TVCF $V^{\mathrm{PE}}$ is nonzero during this period of time. In addition, the trajectories of the robots, represented by the green dashed curves shown in Fig.~\ref{fig:exp3_2000}, indicate that the safety is maintained throughout the experiment.

%% file: FigTab/Fig7.tex
\begin{figure*}[t]
\centering
\subfigure[Time Step $1$]{
\includegraphics[width=0.23\textwidth]{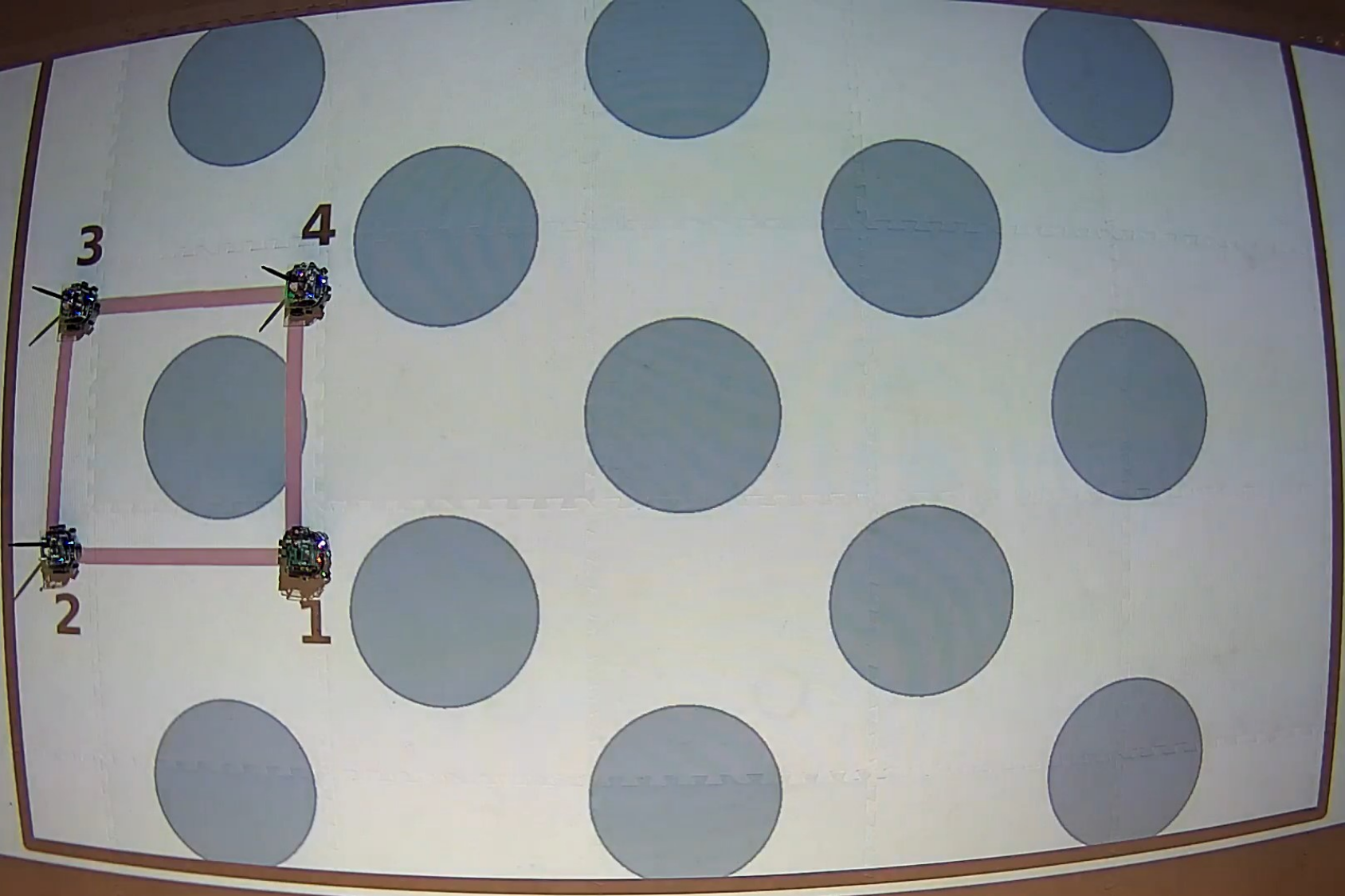}
\label{fig:exp3_1}
}
\subfigure[Time Step $500$]{
\includegraphics[width=0.23\textwidth]{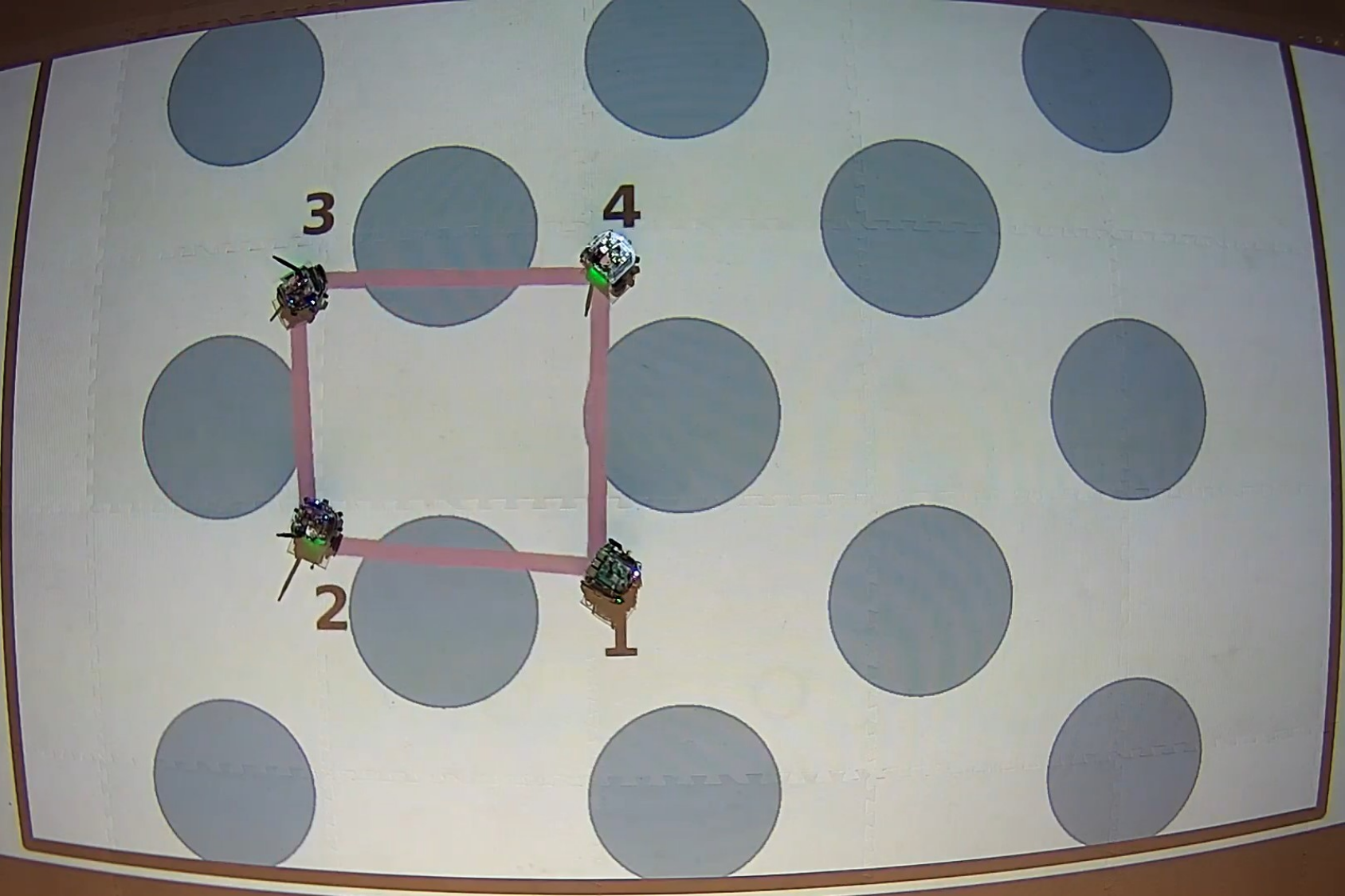}
}
\subfigure[Time Step $1000$]{
\includegraphics[width=0.23\textwidth]{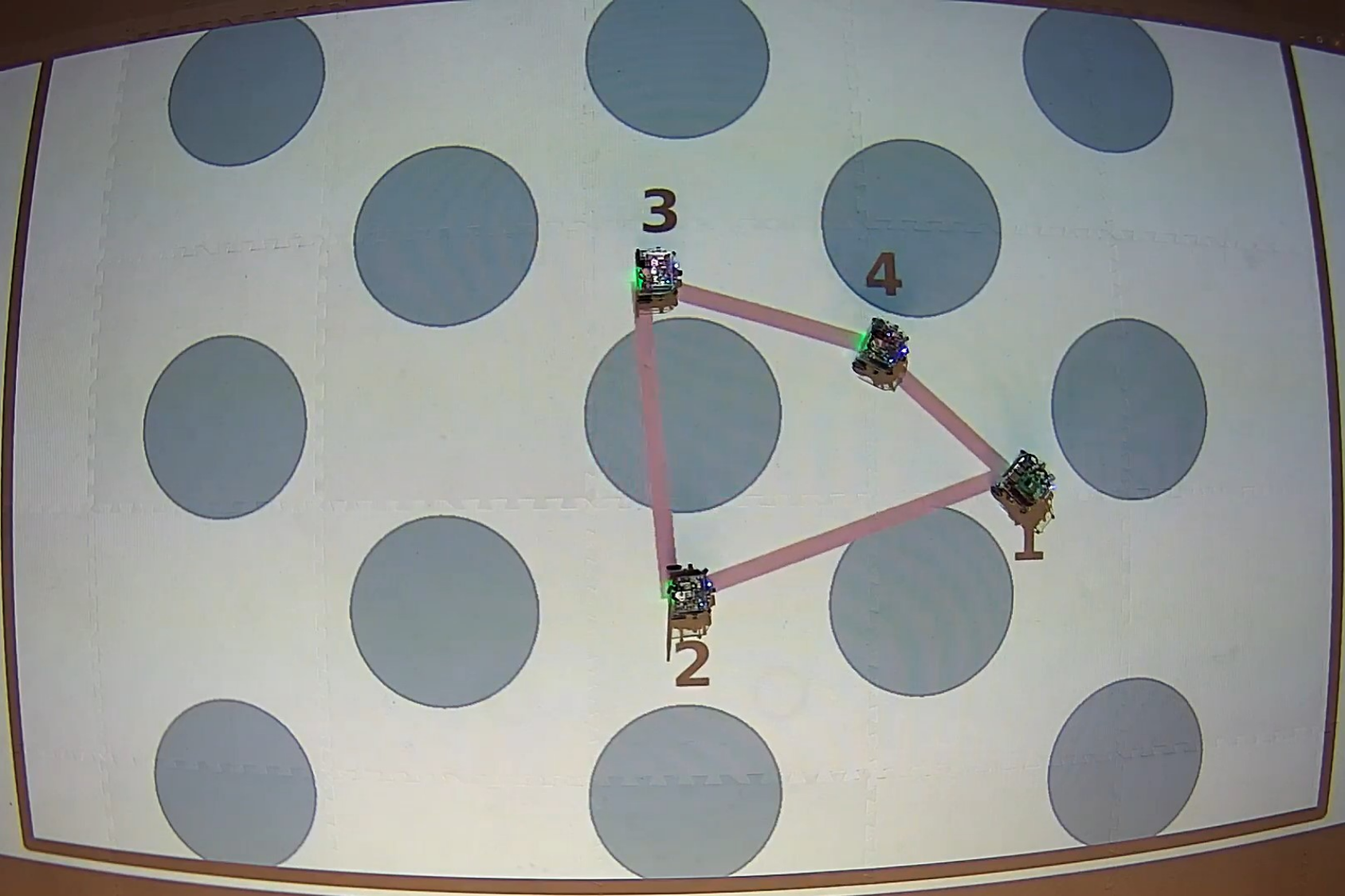}
}
\subfigure[Time Step $2000$]{
\includegraphics[width=0.23\textwidth]{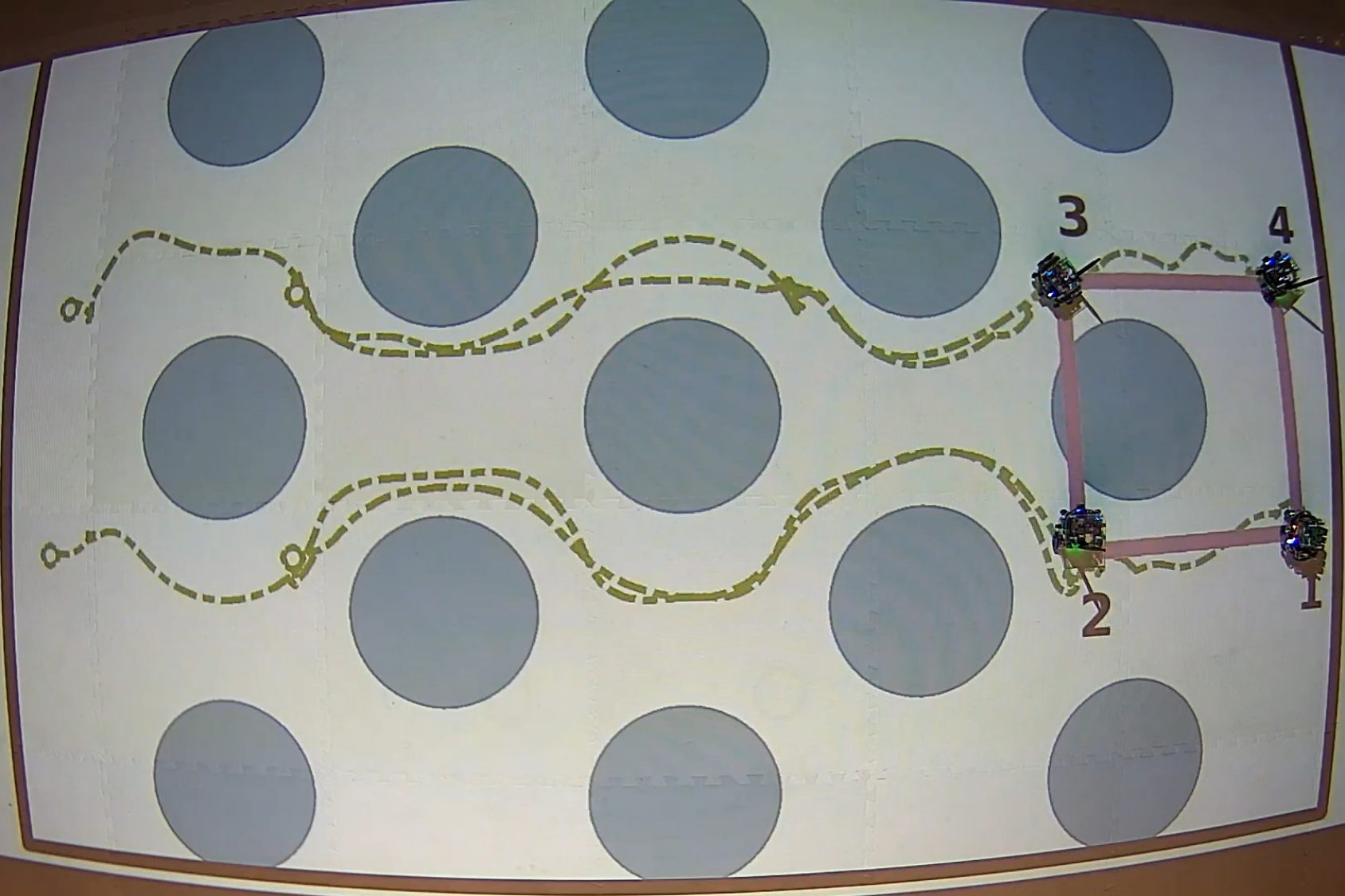}
\label{fig:exp3_2000}
}
\caption{Snapshots of the physical experiment for safe formation navigation in a dense environment in Section~\ref{Section:subsec_exp3}. The magenta lines represent the undirected edges of the corresponding cost graph, the blue disks represent the obstacles, and the green dashed curves represent the trajectories of the robots traveled up to time step $2000$. The full video of this experiment is available online at \url{https://youtu.be/CMFtsuBCy0k}.
}
\label{fig:exp3}
\end{figure*}

%% file: FigTab/Fig8.tex
\begin{figure}[t]
\centering
\includegraphics[scale=0.3]{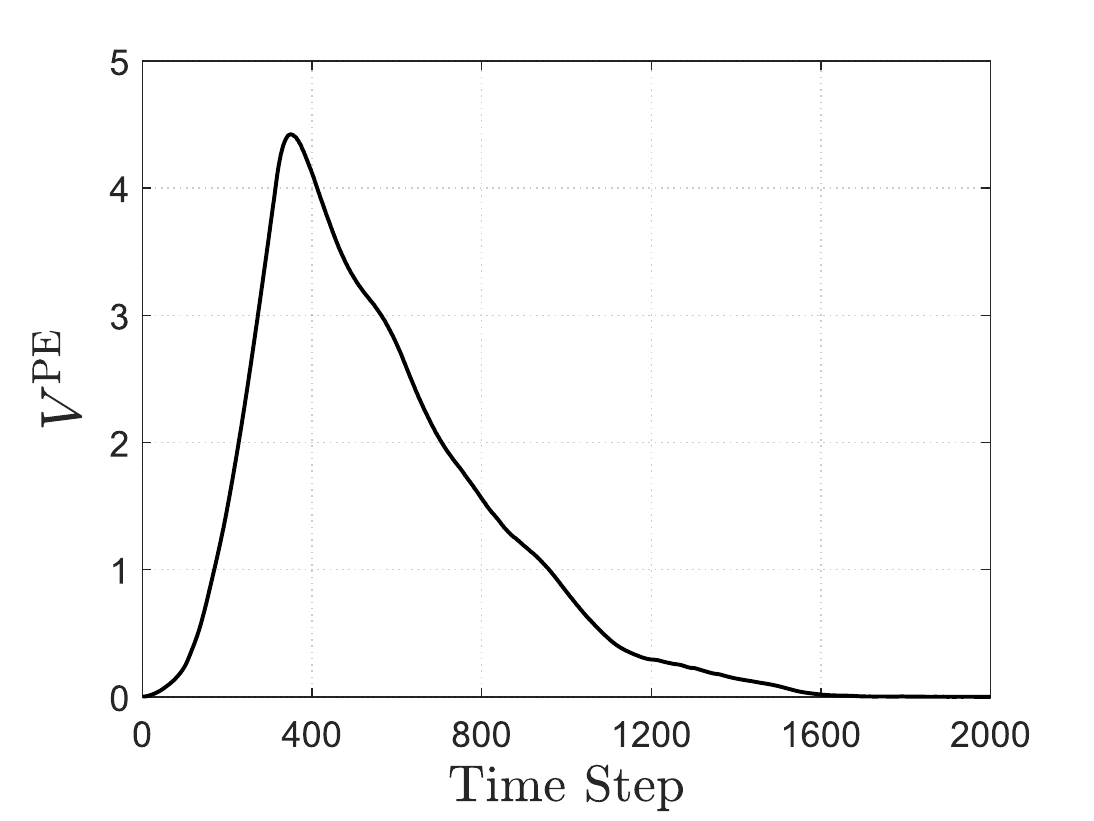}
\caption{The evolution of the total PE-TVCF of the follower robots defined in Section~\ref{Section:subsec_exp3} with respect to time steps.}
\label{fig:exp3_curve}
\end{figure}

%% file: sections/concl.tex
In this paper, we develop a general multi-agent control framework, named {disentangled control}, which provides a systematic approach to decentralized control synthesis without inducing entangled dynamics among the agents. For example, it can be applied to address the long-standing open problem, decentralized coverage control for time-varying density functions without Neumann series approximations. Moreover, the framework is suitable for multi-objective robotics and real-time implementation, as demonstrated in the physical experiments. 

In addition to the applications demonstrated in this paper, disentangled control can be applied to a wide range of problems in a decentralized manner. To name a few, aerial-ground robotic collaboration where the heterogeneous agents have states of different dimensions, i.e., $n_i \neq n_j$, and decentralized coverage control for multiple different time-varying densities, i.e., $\phi_i(q,t) \neq \phi_j(q,t)$. Extending even beyond robotics, the framework holds significant promise for decentralized decision-making and action generation in multi-agent systems. We leave to the community further exploration of the ideas of \eqref{eqn:firstOptimization} and \eqref{eqn:secondOptimization} in broader fields and more diverse applications. Another avenue for future work is to consider systems with higher-order dynamics, more general input constraints, disturbances, and stochastic noise.

%% file: sections/appendix1.tex
In this appendix, we detail several theoretical limitations in the closely related works \cite{notomista2019constraint, santos2019decentralized} while adopting the notation convention of \cite{notomista2019constraint,santos2019decentralized} for clarity.

First, we revisit the constraint-driven multi-robot control framework proposed in \cite{notomista2019constraint}.
\input{theorem/prop1}

Proposition \ref{prop:constraintdriven} is a restatement of Proposition 8 in \cite{notomista2019constraint}. However, the structural assumptions on the cost function \eqref{eqn:paircost} in Proposition \ref{prop:constraintdriven} somewhat restrict the classes of problems that can be covered. For example, Proposition \ref{prop:constraintdriven} assumes that $J_{ij}$ is symmetric and pairwise, but for the problem of leader-follower formation control, even in the time-invariant case, such an assumption does not hold for directed leader-follower interactions where the followers' cost functions depend on the leader state, whereas the leader's cost function does not depend on the follower states.

In addition, Proposition~\ref{prop:constraintdriven} assumes that $J_{ij}$ can be written as an explicit function of the distance between two agents, but such an assumption does not hold for Voronoi-based coverage control, even with time-invariant density functions. Nevertheless, Proposition \ref{prop:constraintdriven} provides convergence guarantees for consensus control (or rendezvous) and formation control in the time-invariant settings.

Next, we clarify several issues in \cite{santos2019decentralized}, which applies \cite{notomista2019constraint} to coverage control for time-varying density functions.

One issue is that when \cite{santos2019decentralized} calculates the total derivatives of some quantities, the influence from the neighbors is missing. For example, Proposition 6 in \cite{santos2019decentralized} calculates the total derivative of $J_i(x,t) = \frac{1}{2}\|x_i - G_i(x,t)\|^2$, where $G_i$ is the center of mass of the Voronoi cell of robot $i$, as
\begin{equation*}
\dot{J}_i(x,t)
=
\frac{\partial J_i}{\partial x_i}
\dot{x}_i
+ 
\frac{\partial J_i}{\partial t},
\end{equation*}
which is inaccurate in the context of coverage control because
\begin{equation*}
\dot{J}_i(x,t)
=
\frac{\partial J_i}{\partial x_i}
\dot{x}_i
+ \!
\sum_{j \in \mathcal{N}_i}
\frac{\partial J_i}{\partial x_j}
\dot{x}_j
+
\frac{\partial J_i}{\partial t},
\end{equation*}
where $\mathcal{N}_i$ denotes the Delaunay neighbors of robot $i$. 

Another issue is that the conclusion of Proposition 4 in \cite{santos2019decentralized} does not necessarily hold (and cannot be proved by following similar to \cite{notomista2019constraint}), not only for coverage control, but also for other multi-robot tasks with even relatively simpler forms of  $J_i(x,t)$. This is because the convergence condition given by Proposition 4 in \cite{santos2019decentralized}, which can be equivalently rewritten as
\begin{equation*}
\label{eqn:santoscond}
\frac{\partial J_i}{\partial x_i}
\dot{x}_i
+ 
\frac{\partial J_i}{\partial t}
\leq 
-\alpha(J_i),
\end{equation*}
where $\alpha$ is a subadditive extended class $\mathcal{K}$ function, does not guarantee the convergence of $J(x,t) = \sum^N_{i=1} J_i(x,t)$.

A further issue is that even in the special case where $J_i(x,t)$ can be written as an explicit function
of the distance between two robots, i.e., 
\begin{equation}
\label{eqn:TVpairwisecost}
J_i(x,t) = 
\sum_{j \in \mathcal{N}_i} J_{ij} (\|x_i - x_j\|,t),
\end{equation}
although such a form \eqref{eqn:TVpairwisecost} is not admissible in coverage control, the convergence of $J(x,t) = \sum^N_{i=1} J_i(x,t)$ still cannot be guaranteed by Proposition 4 in \cite{santos2019decentralized}. 

A special case where the method proposed in \cite{santos2019decentralized} can ensure the convergence of $J$ is when 
\begin{equation*}
\frac{\partial J}{\partial x_i} = \frac{\partial J_i}{\partial x_i},
\end{equation*}
but such a condition is rather restrictive, as it does not hold even in basic multi-robot tasks such as rendezvous.

%% file: theorem/prop1.tex
\begin{proposition}
\label{prop:constraintdriven} 
Given a group of $N$ robots with single-integrator dynamics $\dot{x}_i = u_i$, consider the cost function
\begin{equation}
\label{eqn:paircost}
J(x) = \sum_{i=1}^N \sum_{j \in \mathcal{N}_i} J_{ij} (\|x_i - x_j\|),
\end{equation}
where $\mathcal{N}_i$ denotes the set of indices of the neighbors of robot~$i$, and $J_{ij}: \mathbb{R}_{\geq 0} \to \mathbb{R}_{\geq 0}$, $J_{ij} (\|x_i - x_j\|) = J_{ji} (\|x_j - x_i\|)$ is a symmetric, pairwise
function between robots $i$ and $j$. If each robot executes the controller $u_i$ such that
\begin{equation}
\label{eqn:surrogateconstraint}
\frac{\partial J_i}{\partial x_i}
u_i \leq -\alpha(J_i),
\end{equation}
$\forall t \geq 0$, $\forall i \in \mathcal{N}$, where $J_i(x) = \sum_{j \in \mathcal{N}_i}J_{ij} (\|x_i - x_j\|)$ and $\alpha$ is a subadditive extended class $\mathcal{K}$ function, then $J(x)$ given by \eqref{eqn:paircost} asymptotically converges to $0$.
\end{proposition}